\documentclass[a4paper,USenglish,cleveref, autoref, thm-restate,
numberwithinsect,screen]{lipics-v2021}

\hideLIPIcs

\nolinenumbers

\graphicspath{{./images/}}%

\title{Ranked Enumeration for MSO on Trees via Knowledge Compilation} %

\titlerunning{Ranked Enumeration for MSO on Trees via Knowledge Compilation}

 \author{Antoine Amarilli}{LTCI, Télécom Paris, Institut
 polytechnique de Paris, France \and
\url{https://a3nm.net/}}{antoine.amarilli@telecom-paris.fr}{https://orcid.org/0000-0002-7977-4441}{Partially
supported by the ANR project EQUUS ANR-19-CE48-0019, by the
Deutsche Forschungsgemeinschaft (DFG, German Research Foundation) –
431183758, and by the ANR project ANR-18-CE23-0003-02 (“CQFD”).
This work was done in part while the author was visiting the Simons
Institute for the Theory of Computing. } 

 \author{Pierre Bourhis}{Univ. Lille, CNRS, Inria, Centrale Lille, UMR 9189 CRIStAL, F-59000 Lille, France}{%
   pierre.bourhis@univ-lille.fr
 }{%
   0000-0001-5699-0320
 }{}

 \author{Florent Capelli}{Univ. Artois, CNRS, UMR 8188, Centre de Recherche en Informatique de Lens (CRIL), F-62300 Lens, France
 \and \url{https://florent.capelli.me/}
 }{%
   capelli@cril.fr%
 }{%
   https://orcid.org/0000-0002-2842-8223%
 }{%
   This work was supported by project ANR KCODA, ANR-20-CE48-0004.%
 }
 
 \author{Mikaël Monet}{%
     Université de Lille, CNRS, Inria, UMR 9189 - CRIStAL, F-59000 Lille, France \and \url{https://mikael-monet.net/}%
 }{%
   mikael.monet@inria.fr%
 }{%
   0000-0002-6158-4607
 }{This work was done in part while the author was visiting the Simons Institute for the Theory of Computing.}

\authorrunning{A. Amarilli, P. Bourhis, F. Capelli and M. Monet} %

\Copyright{Antoine Amarilli, Pierre Bourhis, Florent Capelli and Mikaël Monet} %

\ccsdesc[500]{Information systems~Relational database model} %

\keywords{Enumeration, knowledge compilation, monadic second-order logic} %

\category{} %

\relatedversion{} %
\usepackage[utf8]{inputenc}
\usepackage{amsthm}
\usepackage[ruled,vlined, linesnumbered]{algorithm2e}
\usepackage{mathtools}
\usepackage[bibliography=common]{apxproof}
\usepackage{colonequals}

\DeclareMathOperator*{\argmax}{arg\,max}

\NewDocumentCommand{\parta}{+O{X} +O{D}}{\overline{#2^#1}}
\NewDocumentCommand{\lit}{+O{x} +O{d}}{\langle{#1:#2}\rangle}
\NewDocumentCommand{\var}{+m}{\mathsf{var}(#1)}
\NewDocumentCommand{\rel}{+m}{\mathsf{rel}(#1)}

\NewDocumentCommand{\op}{}{\mathsf{op}}
\NewDocumentCommand{\supp}{}{\mathsf{supp}}
\NewDocumentCommand{\dom}{}{\mathsf{dom}}

\NewDocumentCommand{\exit}{+m}{\mathsf{exit}(#1)}

\newcommand\RR{\mathbb{R}}
\newcommand\NN{\mathbb{N}}

\newcommand\pp{\mathbf{p}}
\newcommand\dd{\mathbf{d}}

\newcommand\Enum{\mathsf{RankEnum}}

\newcommand\get{\mathsf{Get}}

\newtheoremrep{theorem}{Theorem}[section]
\newtheoremrep{corollary}[theorem]{Corollary}
\newtheoremrep{proposition}[theorem]{Proposition}
\newtheoremrep{definition}[theorem]{Definition}
\newtheoremrep{lemma}[theorem]{Lemma}
\newtheoremrep{remark}[theorem]{Remark}
\newtheoremrep{conjecture}[theorem]{Conjecture}
\newtheoremrep{claim}[theorem]{Claim}

\newtheorem{result}{Result}

\begin{document}

\maketitle

\begin{abstract}
We study the problem of enumerating the satisfying assignments for certain circuit
classes from knowledge compilation, where assignments are ranked in a specific
order. In particular, we show how this problem can be used to efficiently perform ranked
enumeration of the answers to MSO queries over trees, with the order being given
by a ranking function satisfying a subset-monotonicity property.

Assuming that the 
number of variables is constant, we show that we can enumerate the satisfying
assignments in ranked order for so-called \emph{multivalued circuits} that are smooth,
decomposable, and in negation normal form (smooth multivalued DNNF). There is
no preprocessing and the enumeration delay is 
linear in the size of the circuit times the number of values, plus a logarithmic
term in
the number of assignments produced so far. If we further assume that the
circuit is deterministic (smooth multivalued d-DNNF), we can achieve
linear-time preprocessing in the circuit, and the delay only features the
logarithmic term.

\end{abstract}

\section{Introduction}
\label{sec:introduction}
Data management tasks often require the evaluation of queries on large
datasets, in settings where the number of query answers may be very large. For
this reason, the framework of \emph{enumeration algorithms} has been proposed as
a way to distinguish the \emph{preprocessing time} of query evaluation
algorithms and the maximal \emph{delay} between two successive
answers~\cite{Segoufin14,Wasa16}. Enumeration algorithms have been studied in
several contexts: for conjunctive
queries~\cite{bagan2007acyclic} and unions of conjunctive
queries~\cite{berkholz2018answeringucqs,carmeli2021enumeration} over relational
databases; for first-order logic over bounded-degree
structures~\cite{durand2007first}, structures with local bounded expansion~\cite{segoufin2017constant}, and
nowhere dense graphs~\cite{schweikardt2022enumeration}; and for monadic
second-order logic (MSO) over
trees~\cite{bagan2006mso,kazana2013enumeration,amarilli2017circuit}.

We focus on the setting of MSO over trees. In this context, 
the following enumeration result is already known.
For any fixed MSO query $Q$ (i.e., in \emph{data complexity}) where the 
free variables are assumed to be first-order, considering the 
answers of~$Q$ on a tree~$T$ given as input
(i.e., the functions that map the variables of~$Q$ to nodes of~$T$ in a way that satisfies~$Q$),
we can enumerate them with
linear preprocessing on the tree~$T$ and with constant delay. If the free variables
are second-order, then the delay is output-linear, i.e., linear in each
produced answer~\cite{bagan2006mso,amarilli2017circuit}. Further results are
known when the query is not fixed but given as input as a potentially
non-deterministic automaton~\cite{amarilli2019constant,amarilli2019enumeration},
or when maintaining the enumeration structure under tree updates~\cite{losemann2014mso,amarilli2019enumeration}.

However, despite their favorable delay bounds,
a shortcoming of these enumeration algorithms is that they enumerate answers in
an opaque order which cannot be controlled.
This is in contrast
with application settings where answers should be enumerated, e.g., by
decreasing order of relevance, or focusing on the \emph{top-$k$} 
most relevant answers. This justifies the need for enumeration algorithms
that can produce answers in a user-defined order, even if they do so at the
expense of higher delay bounds.

This task, called \emph{ranked enumeration}, has
recently been studied in various contexts. For instance, Carmeli et
al.~\cite{carmeli2023tractable,bringmann2022tight,DBLP:journals/corr/abs-2201-02401} study for which order functions one
can efficiently perform ranked direct access to the answers
of conjunctive queries: here, efficient ranked direct access implies
efficient ranked enumeration. Ranked enumeration has also been studied to
support order-by
operators on factorized databases~\cite{bakibayev2013aggregation}.
Other works have studied ranked enumeration for document
spanners~\cite{doleschal2022weight}, which relate to the evaluation of MSO
queries over words.
Closer to applications, some works have studied the ranked enumeration of
conjunctive query answers, e.g., Deep et
al.~\cite{deep2019ranked,deep2022ranked} or Tziavelis et
al.~\cite{tziavelis2020optimal,tziavelis2022any}.
Variants of in-order enumeration have been also studied on knowledge
compilation circuit classes, for
instance top-$k$, with a pseudo-polynomial time
algorithm~\cite{bourhis2022pseudo}.
Closest to the present work, Bourhis et al.~\cite{bourhis2021ranked} have
studied enumeration on \emph{words} where the ranking function on answers is
expressed in the formalism of MSO cost functions. They show that enumeration can be performed with linear
preprocessing, with a delay between answers which is no longer constant but
logarithmic in the size of the input word. However, their result does not apply in the more
general context of trees.

\subparagraph*{Contributions.} In this paper, we embark on the study of
efficient ranked enumeration algorithms for the answers to MSO queries on trees,
assuming that all free variables are first-order.
We define this task by assigning \emph{scores} 
to
each so-called \emph{singleton assignment} $[x\to d]$ describing that variable $x$ is
assigned tree node~$d$, and combining these values into a \emph{ranking
function} while assuming a
\emph{subset-monotonicity property}~\cite{tziavelis2022any}: intuitively, when
extending two partial assignments in the same manner, then the order between
them does not change.
This setting covers many ranking functions, e.g.,
those defined by order, sum, or a lexicographic order on the variables.
Our main contribution is then to show the following
results on the data complexity of ranked enumeration for MSO queries on trees:

\begin{result}
  \label{res:mso}
  For any fixed MSO query $Q(x_1, \ldots, x_n)$ with free first-order variables,
  given as input a tree $T$ and a
  subset-monotone ranking function $w$ on the partial assignments of
  $x_1, \ldots, x_n$
  to nodes of~$T$, we can enumerate the answers to~$Q$ on~$T$ in nonincreasing
  order of scores
  according to~$w$
  with a preprocessing time 
  of $O(|T|)$
  and a delay of 
  $O(\log (K+1))$,
  where $K$ is the number of answers produced so far.
\end{result}

Note that, as the total number of answers is at most $|T|^n$,
and as $n$ is constant in data complexity,
the delay of
$O(\log (K+1))$ can alternatively be bounded by $O(n \log |T|)$, or $O(\log |T|)$.
This matches the bound
of~\cite{bourhis2021ranked} on words, though their notion of
rank is different.
Further, our bound
shows that the first answers can be produced faster,
e.g., for top-$k$ computation.

Our results for MSO queries on trees are shown in the general framework of
\emph{circuit-based enumeration methods}, introduced
by~\cite{amarilli2017circuit}. In this framework, enumeration results are
achieved by first translating the task to a class of structured circuits from
knowledge compilation, and then proposing an enumeration algorithm that works
directly on the structured class. This makes it possible to re-use enumeration
algorithms across a variety of problems that compile to circuits. In this paper,
as our task consists in enumerating assignments (from first-order variables of an MSO query to tree nodes),
we phrase our results in terms of \emph{multivalued circuits}. These circuits
generalize Boolean circuits by allowing variables to take values in a larger
domain than $\{0, 1\}$: intuitively, the domain will be the set of the tree nodes. We assume that
circuits are \emph{decomposable}, i.e., that no variable has a path to two
different inputs of a $\land$-gate: this yields \emph{multivalued DNNFs}, which
generalize usual DNNFs. We also assume that the circuits are \emph{smooth}:
intuitively, no variable is omitted when combining partial assignments at
an~$\lor$-gate. Multivalued circuits can be smoothed while preserving
decomposability, in quadratic time or faster in some
cases~\cite{shih2019smoothing}. Smooth multivalued DNNF circuits can alternatively be understood
as factorized databases, but we do not impose that they
are \emph{normal}~\cite{olteanu2015size}, i.e., the depth can be arbitrary.

Our enumeration task for MSO on trees thus amounts to the enumeration of
satisfying assignments of smooth multivalued DNNFs, following a ranking function which we
assume to be subset-monotone. However, we are not aware of existing results 
for ranked enumeration on circuits in the knowledge compilation literature.
For this reason, the second contribution of
this paper is to show efficient enumeration algorithms on these
smooth multivalued DNNFs.

We first present an algorithm for this task that runs with no preprocessing and
polynomial delay.
The algorithm can be seen as an instance of the
Lawler-Murty~\cite{lawler1972procedure,murty1968algorithm} procedure. We show:

\begin{result}
  \label{res:dnnf}
  For any constant $n \in \NN$,
  given a smooth multivalued DNNF circuit $C$ with domain~$D$ and with
  $n$ variables, given
  a subset-monotone ranking function $w$, we can enumerate the
  satisfying assignments of~$C$ in nonincreasing order of scores according to~$w$
  with delay $O(|D|\times |C| + \log (K+1))$, where~$K$ is the number of assignments produced so far.
\end{result}

We then show a second algorithm, which allows for a better delay bound at the expense of making
an additional assumption on the circuit; it is with this algorithm that we prove Result~\ref{res:mso}.
The additional assumption is that the circuit is \emph{deterministic}:
intuitively, no partial assignment is captured twice. This corresponds to the
class of \emph{smooth multivalued d-DNNF circuits}. For our task of enumerating
MSO query answers, the determinism
property can intuitively be enforced on circuits when we compute them using 
an deterministic tree automaton to represent the query. We then show:

\begin{result}
  \label{res:ddnnf}
  For any constant $n \in \NN$,
  given a smooth multivalued d-DNNF circuit $C$ with $n$ variables,
  given a subset-monotone ranking function $w$, we can enumerate the
  satisfying assignments of~$C$ in nonincreasing order of scores according to~$w$
  with preprocessing time $O(|C|)$ and delay $O(\log (K+1))$,
  where $K$ is the number of assignments produced so far.
\end{result}

\subparagraph*{Paper structure.}
We give preliminary definitions in Section~\ref{sec:preliminaries}. We first
study in Section~\ref{sec:DNNF} the ranked enumeration problem for smooth
multivalued DNNF circuits (Result~\ref{res:dnnf}). We then move on to a more
efficient algorithm on smooth multivalued d-DNNF circuits
(Result~\ref{res:ddnnf}) in Section~\ref{sec:d-DNNF}. We show how to apply the
second algorithm to ranked enumeration for the answers to MSO queries
(Result~\ref{res:mso}) in Section~\ref{sec:mso}. We conclude in
Section~\ref{sec:conclusion}.

\section{Preliminaries}
\label{sec:preliminaries}
\begin{toappendix}
  \label{apx:prelims}
\end{toappendix}
For~$n\in \NN$ we write~$[n]$ the set~$\{1,\ldots,n\}$.
\subparagraph{Assignments.} For two finite sets $D$ of \emph{values}
and $X$ of \emph{variables}, an \emph{assignment on domain~$D$ and variables
$X$} is a mapping from $X$ to $D$. We write $D^X$ the set of such
assignments. 
We can see assignments as sets of \emph{singleton assignments}, where a
\emph{singleton assignment} is an expression of
the form $[x\to d]$ with $x \in X$ and $d \in D$.

Two assignments $\tau \in D^Y$ and $\sigma \in
D^Z$ are \emph{compatible}, written~$\tau \simeq \sigma$,
if we have $\tau(x)=\sigma(x)$ for every~$x\in Y\cap Z$.
In this case, we denote by $\tau \bowtie
\sigma$ the assignment of $D^{Y \cup Z}$ defined following the natural join,
i.e., for~$y\in Y\setminus Z$ we set $(\tau
\bowtie \sigma)(y) \colonequals \tau(y)$, for~$z\in Z\setminus Y$ we set $(\tau
\bowtie \sigma)(z) \colonequals
\sigma(z)$, and for $x\in Z\cap Y$, we set $(\tau \bowtie
\sigma)(x)$ 
to the common value $\tau(x) = \sigma(x)$.
Two assignments $\tau \in
D^Y$ and $\sigma \in D^Z$ are \emph{disjoint}
if $Y \cap Z = \emptyset$: then 
they are always compatible and $\tau \bowtie \sigma$
corresponds to the relational product, which we write $\tau \times \sigma$.

Given $R \subseteq D^Y$ and $S \subseteq D^Z$,
we define
$R \wedge S = \{\tau \bowtie \sigma \mid \tau \in R,\, \sigma \in S,\, \tau \simeq
\sigma\}$: this is a subset of $D^{Y \cup Z}$. Note how, if the domain is $D = \{0, 1\}$, then this corresponds to
the usual conjunction for Boolean functions, and in general we can see it as a relational join, or a relational
product whenever $Y \cap Z = \emptyset$.
Further, we define $R \vee S = \{\tau
\in D^{Y \cup Z} \mid \tau|_Y \in R \text{ or } \tau|_Z \in S\}$, which is
again a subset of $D^{Y \cup Z}$.
Again observe how, when $D = \{0, 1\}$, this
corresponds to disjunction; and
in general we can see this as relational union
except that assignments over $Y$ and $Z$ are each implicitly
completed in all possible ways to assignments over $Y \cup Z$.

\subparagraph{Multivalued circuits.}
A \emph{multivalued circuit $C$ on domain 
$D$ and variables $X$} is a DAG with labeled vertices which are called
\emph{gates}. The circuit also has a distinguished gate~$r$ called
the \emph{output gate of $C$}. Gates having no incoming edges are called 
\emph{inputs} of $C$. Moreover, we have:
\begin{itemize}
\item Every input of $D$ is labeled with a pair of the form $\lit[x][d]$ with $x \in X$
  and $d \in D$;
\item Every other gate of $D$ is labeled with either $\vee$ (a
  \emph{$\vee$-gate}) or $\wedge$ (a \emph{$\wedge$-gate}).
\end{itemize}
We denote by $|C|$ the number of edges in 
$C$.

Given a gate $v$ of $C$, the \emph{inputs of $v$} are the gates $w$
of $C$ such that there is a directed edge from $w$ to $v$. The
\emph{set of variables below $v$}, denoted by $\var{v}$, is then the set of
variables $x \in X$ such that there is an input $w$ which is labeled by $\lit$ for some
$d\in D$ and which has a directed path to~$v$. Equivalently, if $v$ is an input
labeled by $\lit$ then $\var{v} \colonequals \{x\}$, otherwise $\var{v}
\colonequals \bigcup_{i=1}^k
\var{v_i}$ where $v_1,\dots,v_k$ are the inputs of $v$. We assume that the set
$X$ of variables of the circuit is equal to $\var{r}$ for $r$ the output gate of~$C$:
this can be enforced without loss of generality up to removing useless
variables from~$X$.

For each gate $v$ of $C$, the \emph{set of assignments $\rel v \subseteq D^{\var v}$
of $v$} is defined inductively as follows. If $v$ is an input labeled
by $\lit$, then $\rel v$ contains only the assignment $[x \mapsto d]$.
Otherwise, if $v$ is an internal gate with inputs $v_1, \dots, v_k$ then $\rel
v \colonequals \rel{v_1}~\op \cdots \op~\rel{v_k}$ where $\op \in \{\vee, \wedge\}$ is the
label of $v$. The \emph{set of assignments} $\rel{C}$ of~$C$ is that of its output gate.
Note that, if $D = \{0, 1\}$, then the set of assignments of~$C$ precisely
corresponds to its satisfying valuations when we see~$C$ as a Boolean circuit in
the usual sense. 

We say that a $\wedge$-gate $v$ is \emph{decomposable} if all its inputs are on
disjoint sets of variables; formally, for every pair of
inputs $v_1 \neq v_2$ of $v$, we have $\var{v_1} \cap \var{v_2} = \emptyset$. A
$\vee$-gate $v$ is \emph{smooth} if all its inputs have the same set of
variables (so that implicit completion does not occur); formally, for every pair of inputs $v_1,v_2$ of $v$,
we have $\var{v_1} = \var{v_2}$. A $\vee$-gate $v$ is \emph{deterministic} if
every assignment of~$v$ is computed by only one of its inputs;
formally,
for every pair of inputs $v_1 \neq v_2$ of $v$, if $\tau \in \rel v$ then either
$\tau|_{\var{v_1}} \notin \rel{v_1}$ or $\tau|_{\var{v_2}} \notin \rel{v_2}$.

Let $v$ be an internal gate with inputs $v_1,\dots,v_k$. Observe that if $v$ is
decomposable, then $\rel v = \bigtimes_{i=1}^k \rel{v_i}$. If $v$ is smooth
then $\rel v = \bigcup_{i=1}^k \rel{v_i}$. If moreover $v$ is deterministic,
then $\rel v = \biguplus_{i=1}^k \rel{v_i}$, where $\uplus$ denotes disjoint
union. Accordingly, we denote decomposable $\wedge$-nodes as
\emph{$\times$-nodes}, denote smooth $\lor$-nodes as \emph{$\cup$-nodes}, and
denote smooth deterministic $\lor$-nodes as \emph{$\uplus$-nodes}.

A multivalued circuit is \emph{decomposable} (resp., \emph{smooth},
\emph{deterministic}) if every $\land$-gate is decomposable (resp., every
$\lor$-gate is smooth, every $\lor$-gate is deterministic).
A \emph{multivalued DNNF on domain $D$ and
variables $X$} is then a decomposable multivalued circuit on~$D$ and~$X$.
A \emph{multivalued d-DNNF on domain $D$ and variables $X$} is a
deterministic multivalued DNNF on~$D$ and~$X$.
In all this paper, we only work with circuits that are both decomposable and
smooth, i.e., smooth multivalued DNNFs. Note that smoothness
can be ensured on Boolean circuits in quadratic time~\cite{shih2019smoothing},
and the same can be done on multivalued circuits.

\subparagraph{Ranking functions.}
Our notion of ranking functions will give a score to each assignment, but 
to state their properties we define them on partial
assignments. Formally, a \emph{partial assignment} is a mapping~$\nu:X\to D\cup
\{\bot\}$, where $\bot$ is a fresh symbol representing \emph{undefined}.
We denote by $\parta$ the set of partial assignments on domain $D$ and
variables $X$. 
The
\emph{support} $\supp(\nu)$ of~$\nu$ is the subset of $X$ on which~$\nu$ is
defined. We extend the definitions of compatibility, of~$\bowtie$, and of
disjointness, to partial assignments in the following way. 
  Two partial assignments $\tau \in \overline{D^Y}$ and $\sigma \in
\overline{D^Z}$ are \emph{compatible}, again written~$\tau \simeq \sigma$, when
for every~$x\in Y\cap Z$, if~$\tau(x)\neq \bot$ and~$\sigma(x)\neq \bot$
then~$\tau(x)=\sigma(x)$. In this case, we denote by $\tau \bowtie \sigma$ the
partial assignment of $\overline{D^{Y \cup Z}}$ defined by: for~$y\in Y\setminus Z$ we
have $(\tau \bowtie \sigma)(y) \colonequals \tau(y)$, for~$z\in Z\setminus Y$
we have $(\tau \bowtie \sigma)(z) \colonequals \sigma(z)$, and for $x\in Z\cap
Y$, if~$\tau(x)\neq \bot$ then $(\tau \bowtie \sigma)(x)=\tau(x)$, otherwise
$(\tau \bowtie \sigma)(x)=\sigma(x)$. We call~$\tau$ and~$\sigma$ 
\emph{disjoint}
if~$Y \cap Z = \emptyset$; then again they are always compatible and we write
$\tau \times \sigma$ for $\tau \bowtie \sigma$.

We then consider ranking functions defined on partial
assignments $\parta$, on which we will impose 
\emph{subset-monotonicity}.
Formally, a \emph{$(D,X)$-ranking function} $w$ 
is a function\footnote{As usual, when we write $\RR$, we assume a suitable
representation, e.g., as floating-point numbers.}
$\parta \to \RR$ that gives a score to every partial assignment.
Such a ranking function induces a \emph{weak ordering\footnote{Recall that a weak ordering~$\preceq$ on~$A$ is a total
  preorder on~$A$, i.e., $\preceq$ is transitive and we have either $x \preceq
  y$ or $y \preceq x$ for every $x,y\in A$. In
  particular, it can be the case that two distinct elements $x$ and $y$ are tied, i.e., $x \preceq y$ and $y
  \preceq x$.}~$\preceq$}
on~$\parta$, with~$\mu \preceq \mu'$ defined as~$w(\mu) \leq
w(\mu')$. We always assume
  that ranking functions can be computed efficiently, i.e., with
  running time that only depends on~$X$, not~$D$.
  By a slight notational abuse, we define the score $w(\tau)$ of
  partial assignment~$\tau\in \overline{D^Y}$ with~$Y\subseteq X$
  by seeing $\tau$ as a partial assignment on~$X$ which is implicitly extended
  by assigning $\bot$ to every $z \in X\setminus Y$.
  Following earlier work~\cite{deep2019ranked,tziavelis2022any,deep2022ranked},
  we then restrict our study to ranking functions that are
  \emph{subset-monotone}~\cite{tziavelis2022any}:
  \begin{definition}
    \label{def:subset-monotone}
  A $(D,X)$-ranking function $w: \parta \to \RR$ is \emph{subset-monotone} if
  for every~$Y\subseteq X$ and partial assignments $\tau_1, \tau_2 \in
  \overline{D^{Y}}$ such that $w(\tau_1) \leq w(\tau_2)$, for every partial
  assignment $\sigma \in \overline{D^{X\setminus Y}}$ (so disjoint with
  $\tau_1$ and $\tau_2$), we have $w(\sigma \times \tau_1) \leq w(\sigma \times
  \tau_2)$. 
\end{definition}
  We use in particular the following consequence of subset-monotonicity (see
  Appendix~\ref{apx:prelims}), where we call~$\tau \in \parta$ \emph{maximal}
  (or maximum) for~$w :\parta \to \RR$ when for every~$\tau' \in \parta$ we
  have~$w(\tau') \leq \tau(\tau)$:

\begin{lemmarep}
  \label{lemma:maxdnnf} Let $R\subseteq \overline{D^Y}$ and $S \subseteq \overline{D^Z}$ with $Y \cap Z
  = \emptyset$, and let $w: \overline{D^{Y\cup Z}} \to \RR$ be subset-monotone.
  If $\tau$ is a maximal element of $R$ and $\sigma$ is a maximal element of
  $S$ with respect to $w$, then $\tau \times \sigma$ is a maximal element of $R
  \wedge S$ with respect to $w$.
\end{lemmarep}

\begin{toappendix}
  Lemma~\ref{lemma:maxdnnf} is a direct consequence of the following lemma:
\begin{lemma}
  \label{lem:subsetmonotone} A $(D,X)$-ranking function $w$ is subset-monotone
  if and only if the following holds: for every partial assignments $\tau_1,\tau_2$ and $\sigma_1,
  \sigma_2$,
  if 
  $\tau_i$ is disjoint from $\sigma_i$ for each $i \in \{1, 2\}$ and $w(\tau_1) \leq w(\tau_2)$ and
  $w(\sigma_1) \leq w(\sigma_2)$, then we have $w(\tau_1\times \sigma_1) \leq
  w(\tau_2 \times \sigma_2)$.
\end{lemma}
\begin{proof}
  The right to left implication is obvious as we can apply the property with
  $\sigma_1=\sigma_2$ to recover the definition of subset-monotone. For the
  other way around, assume $w$ to be subset-monotone and let
  $\tau_1,\tau_2,\sigma_1,\sigma_2$ be partial assignments as in the statement. By subset-monotonicity, first by applying the
  property with $\sigma=\sigma_1$ and then with $\sigma=\tau_2$, and using
  transitivity, we have
  $w(\tau_1 \times \sigma_1) \leq w(\tau_2 \times \sigma_1) \leq w(\tau_2 \times
  \sigma_2)$. 
\end{proof}

We point out here that the converse of Lemma~\ref{lemma:maxdnnf} is not
necessarily true, i.e., it is not always the case that every maximal assignment~$\tau$ of
a subset-monotone ranking function~$w: \overline{D^{Y\cup Z}} \to \RR$
with~$Y\cap Z=\emptyset$ can be decomposed as $\tau = \tau_1 \times \tau_2$
with~$\tau_1$ some maximal assignment of~$\overline{D^Y}$ and~$\tau_2$ some
maximal assignment of~$\overline{D^Z}$. We show an example of this phenomenon below.

\begin{example}
  \label{expl:subset-monotonicity-is-weird}
  Take~$D=\{a\}$,~$Y=\{y\}$, $Z=\{z\}$ with~$y\neq z$,
  and consider the ranking function~$w: \overline{D^{\{y,z\}}} \to \RR$ defined by
\begin{equation}
  w(\sigma) =
    \begin{cases}
      0 & \text{if } \sigma(y)=\sigma(z)=\bot\\
      50 & \text{if } \sigma(y) = \bot \text{ and } \sigma(z)=a\\
      100 & \text{if } \sigma(y) = a \text{ and } \sigma(z)=\bot\\
      100 & \text{if } \sigma(y) = a \text{ and } \sigma(z)=a
    \end{cases}.       
\end{equation}
  The reader is invited to check that~$w$ is subset-monotone but does not
  satisfy the property mentioned above (take~$\tau = [y \to a, z\to \bot]$).
\end{example}
\end{toappendix}
We give a few examples of subset-monotone ranking functions.
Let~$W:X\times D \to \RR$ be a function assigning scores to
singleton assignments, and define the $(D,X)$-ranking
function~$\mathsf{sum}_W\colon \parta \to \RR$ by $\mathsf{sum}_W(\tau) =
\sum_{x\in X, \tau(x) \neq \bot} W(x,\tau(x))$. Then~$\mathsf{sum}_W$
is subset-monotone. Similarly define $\mathsf{max}_W\colon \parta \to \RR$ by
$\mathsf{max}_W(\tau) = \max_{x\in X, \tau(x) \neq \bot}
W(x,\tau(x))$, 
or $\mathsf{prod}_W$ in a similar manner (with non-negative scores for
singletons); then these are again subset-monotone. In particular, we can use
$\mathsf{sum}_W$ 
to encode lexicographic orderings on~$\parta$.

\subparagraph*{Enumeration and problem statement.}
Our goal in this article is to efficiently enumerate the satisfying assignments
of circuits in nonincreasing order according to a ranking function. We will in
particular apply this for the ranked enumeration of the answers to MSO queries
on trees, as we will explain in Section~\ref{sec:mso}. We call this problem
$\Enum$.
Formally, the input to $\Enum$ consists of a multivalued circuit
$C$ on domain $D$ and variables $X$, and a $(D,X)$-ranking function~$w$ that is subset-monotone.
The output to enumerate
consists of all of~$\rel{C}$, without duplicates, in nonincreasing order of
scores (with ties broken arbitrarily).

Formally, we work in the RAM model on words of
logarithmic size~\cite{aho1974design}, where memory cells can represent
integers of value polynomial in the input length, and on which arithmetic
operations take constant time. We will in particular allocate arrays of
polynomial size in constant time, using lazy initialization~\cite{grandjean2022arithmetic}.
We measure the performance of our algorithms in the framework of
\emph{enumeration algorithms},
where we distinguish two phases.
First, in the \emph{preprocessing phase}, the algorithm reads the input and builds
internal data structures. We measure the running time of this phase as a
function of the input; in general the best possible bound is \emph{linear
preprocessing}, 
e.g.,
preprocessing in $O(|C|)$. Second, in
the \emph{enumeration phase}, the algorithm produces the assignments, one after
the other, without duplicates, and in nonincreasing order of scores; the order of assignments
that are tied according to the ranking function is not specified.
The \emph{delay} is the maximal time that the enumeration phase can take to
produce the next assignment, or to conclude that none are left.
We measure the
delay as a function of the input, as a function of the
produced assignments (which each have size $|X|$), and also as a function of the
number of results that have been produced so far. The best delay is
\emph{output-linear delay}, i.e., $O(|X|)$, which can be achieved for 
(non-ranked) enumeration of MSO queries on
trees~\cite{bagan2006mso,kazana2013enumeration,amarilli2017circuit}.
In our results, we will always fix~$|X|$ to a constant (for technical reasons
explained in the next section), so the corresponding bound would be
\emph{constant delay}, but, like~\cite{bourhis2021ranked}, we will not be able
to achieve it. Also note that the memory usage of the enumeration phase is not
bounded by the delay, but can grow as enumeration progresses.

\subparagraph{Brodal queues.} Similar to~\cite{bourhis2021ranked}, our
algorithms in this paper will use priority queues, in a specific implementation
called a 
\emph{(functional) Brodal queue}~\cite{brodal1996optimal}. Intuitively, Brodal
queues are priority queues which support union operations in $O(1)$, and which
are \emph{purely functional} in the sense that operations return a
queue without destroying the input queue(s). More precisely, a \emph{Brodal
queue} is a data structure which stores a set of priority-data pairs of the form $(\pp: \mathrm{foo}, \dd: \mathrm{bar})$ where~$\mathrm{foo}$ is
a real number and~$\mathrm{bar}$ an arbitrary piece of data, supporting
operations defined below. Brodal queues are
\emph{purely functional and persistent}, i.e., for any operation applied to
some input Brodal queues, we obtain as output a new Brodal queue $Q'$, such
that the input queues can still be used. Note that
the structures of $Q'$ and of the input Brodal queues may be sharing
locations in memory; this is in fact necessary, e.g., to guarantee
constant-time bounds. However, this is done transparently, and both $Q'$ and the
input Brodal queues can be used afterwards\footnote{This is similar to how persistent
  linked lists can be modified by removing the head element or concatenating
  with a new head element. Such operations can run in constant time and return
  the modified version of the list without invalidating the original list; with
both lists sharing some memory locations in a transparent fashion.}. Brodal
queues support the following:
\begin{itemize}
    \item \emph{Initialize}, in time $O(1)$, which produces an empty queue;
    \item \emph{Push}, in time $O(1)$, which adds to~$Q$ a priority-data pair;
    \item \emph{Find-Max}, in time $O(1)$, which either indicates that
      $Q$ is empty or otherwise returns some pair $(\pp: \mathrm{foo}, \dd: \mathrm{bar})$ with
      $\mathrm{foo}$ being maximal among the priority-data pairs stored in $Q$
      (ties are broken arbitrarily);
    \item \emph{Pop-Max}, in time $O(\log(|Q|))$, which either indicates that
      $Q$ is empty or returns two values: first the pair~$p$ returned by
      Find-Max, second a queue
      storing all the pairs of~$Q$ except~$p$;
    \item \emph{Union}, in time $O(1)$, which takes as input a second Brodal
      queue $Q'$ and returns a queue over the elements of $Q$ and $Q'$. 
\end{itemize}

\section{Ranked Enumeration for Smooth Multivalued DNNFs}
\label{sec:DNNF}
\begin{toappendix}
  \label{apx:DNNF}
\end{toappendix}
In this section, we start the presentation of our technical results by giving
our
algorithm to solve the ranked enumeration
problem for DNNFs under subset-monotone orders. 
This is Result~\ref{res:dnnf} from the introduction, which we restate below:

\begin{theorem}
  \label{thm:dnnf}
  For any constant~$n\in \NN$, we can solve the $\Enum$ problem on an input smooth multivalued DNNF
  circuit $C$ on domain~$D$ and variables $X$ with $|X| = n$ 
  and a subset-monotone $(D,X)$-ranking function
  with no preprocessing and with delay 
  $O(|D|\times |C| + \log (K+1))$,
  where~$K$ is the number of assignments produced so far.
\end{theorem}
Note how the number $n$ of variables is assumed to be constant in the result
statement. This is for a technical reason: we will need to store partial
assignments in memory, but in the RAM model we can only index polynomially many
memory locations~\cite[page 3]{grandjean2022arithmetic}, so we must ensure that the
total number of assignments is polynomial. The circuit itself and the domain
can however be arbitrarily large, following the application to
MSO queries over trees studied in
Section~\ref{sec:mso}: the variables of the circuit will be the variables of
the MSO query (which is fixed because we will work in data complexity), and
the size of the circuit and that of the domain will be linear in the size of
the tree (which represents the data).

Our algorithm can be seen as an instance of the
Lawler-Murty~\cite{lawler1972procedure,murty1968algorithm} procedure, that has
been previously used to enumerate paths in DAGs in decreasing order of
weight in~\cite{tziavelis2022any}. Interestingly, the result does not require
that the input circuit is deterministic. However, it is less efficient than
the method presented in Section~\ref{sec:d-DNNF} where determinism is
exploited.

We prove Theorem~\ref{thm:dnnf} in the rest of this section.
Let us fix a smooth multivalued DNNF~$C$ on domain $D$ and variables $X$,
and a subset-monotone ranking function $w \colon \parta \rightarrow \RR$.  For
a partial assignment $\tau$, we denote by $w_C(\tau) = \max \{ w(\tau \times
\sigma) \mid \sigma \in D^{X \setminus \supp(\tau)} \text{ and } \tau \times
\sigma \in \rel{C}\}$ the score of the maximal
completion of $\tau$ to a satisfying assignment of $C$ if it exists and $w_C(\tau) = \bot$
if no such completion exists. Our algorithm relies on the following
folklore observation:

\begin{lemmarep}
  \label{lem:getwc} Given a partial assignment $\tau$, one can compute $w_C(\tau)$ in time $O(|C|)$.
\end{lemmarep}
\begin{proofsketch}
  We condition~$C$ on~$\tau$ in linear time, obtaining~$C'$, and then compute
  in a bottom-up manner for every gate~$C'$ some assignment of~$\rel{g}$ of
  maximal score, relying on smoothness and decomposability. See
  Appendix~\ref{apx:DNNF} for more details.
\end{proofsketch}
\begin{proof}
  Let $X$ be the variables of $C$.

  It is enough to show that we can compute, given a smooth multivalued DNNF~$C'$ and
  monotone ranking function $w'$, some~$\sigma' \in \rel{C'}$ that maximizes
  $w'(\sigma')$, in $O(|C'|)$. Indeed, if this is the case we can first compute
  the conditioning\footnote{See~\cite[Definition 5.4]{darwiche2002knowledge} for the
  definition of conditioning on Boolean circuits, which easily adapts
to multivalued circuits.} $C'$ of $C$ on $\tau$ in time $O(|C'|)$: specifically, $C'$ is a multivalued circuit on domain $D$ and variables $X\setminus \supp(\tau)$ such that, for $\sigma' \in D^{X\setminus \supp(\tau)}$ we have that $\sigma' \in \rel{C'}$ iff $\tau \times \sigma' \in \rel{C}$.
Then, letting~$w'$ be the ranking function on~$\overline{D^{X \setminus \supp(\tau)}}$ defined by
$w'(\sigma') \colonequals w(\sigma' \times \tau)$ (which is subset-monotone),
find one such $\sigma' \in \rel{C'}$ in time~$O(|C|)$, and then return~$w(\sigma'
\times \tau)$. This is correct thanks to subset-monotonicity of~$w$, more precisely, by Lemma~\ref{lemma:maxdnnf}.

  Now the algorithm to do this proceeds by bottom-up induction as follows: for each gate
  $v$ of $C'$, we compute $\sigma_v \in \rel v$ such that $w'(\sigma_v) = \max\{
  w'(\sigma) \mid \sigma \in \rel v \}$. If $v$ is an input then $\rel v$ is a
  singleton assignment, and we let $\sigma_v$ be this assignment. Now, if $v$
  is a $\times$-gate with inputs $v_1, \dots, v_k$, we let $\sigma_v=
  \sigma_{v_1} \times \dots \times \sigma_{v_k}$. By Lemma~\ref{lemma:maxdnnf}, $\sigma_v$ is maximal for
  $\rel{v}$ if each $\sigma_{v_i}$ is maximal for $\rel{v_i}$ which is the case
  by induction. Finally, if $v$ is a $\cup$-gate with input
  $v_1,\dots,v_k$, we define $\sigma_v = \argmax_{i=1}^k w'(\sigma_{v_i})$,
  which is clearly maximal in $\rel v = \bigcup_{i=1}^k \rel{v_i}$ if
  $\sigma_{v_i}$ is maximal in $\rel{v_i}$ for each $i$ because~$v$ is smooth,
  which is the case by induction.
\end{proof}

With this in place, we are ready to describe the algorithm. Notice that our
definition of multivalued circuits implies that $\rel{C}$ can never be empty,
because all gates except input gates have inputs, and the circuit is
decomposable.
We fix an arbitrary order on $X =
\{x_1,\dots, x_n\}$ and, for~$i\in \{1,\ldots,n+1\}$, we denote by $X_{< i}$ the
set $\{x_1,\dots,x_{i-1}\}$ (which is empty for~$i=1$). A partial assignment
$\tau \in \overline{D^{X}}$ is called a \emph{prefix assignment} if
$\supp(\tau) = X_{< i}$ for some $i \in \{1, \ldots, n+1\}$.

\begin{algorithm}[t]
\SetAlgoLined
\KwData{Smooth multivalued DNNF $C$ with $n$ variables, subset-monotone
ranking function $w$.}
\KwResult{Enumeration of the satisfying assignments of~$C$ in nonincreasing order of scores by~$w$.}
  $Q \gets \text{empty priority queue}$\;
  Push the empty assignment $[]$ into $Q$ with priority~$w_C([])$\;
\While{$Q$ is not empty}{
    Pop into $\gamma$ the assignment with maximum $w_C$-score from $Q$\;
    \For{$j \gets |\supp(\gamma)|+1$ to $n$}{
        \ForEach{$d \in D$}{
            Construct $\alpha_d = \gamma \times \langle x_j : d \rangle$\;
            Compute $w_C(\alpha_d)$ using Lemma~\ref{lem:getwc}\;
        }
        $\gamma \gets \alpha_{d_0}$ such that $w_C(\alpha_{d_0})$ is not $\bot$ and is
        maximal\;
        Push into $Q$ all $\alpha_{d'}$ for $d' \neq d_0$ where
        $w_C(\alpha_{d'}) \neq \bot$, with priority $w_C(\alpha_{d'})$\;
    }
  Output $\gamma$\;
}
  \caption{Algorithm for Theorem~\ref{thm:dnnf}}
  \label{alg:dnnf}
\end{algorithm}

The enumeration algorithm is then illustrated as Algorithm~\ref{alg:dnnf}, which
we paraphrase in text below.
The algorithm uses a variable~$\gamma$ holding a prefix assignment
and a priority queue~$Q$ containing prefix assignments. The priorities in the
queue are 
the~$w_C$-score, i.e., the priority of each prefix assignment is the score
returned by~$w_C$ on this assignment.
We initialize~$Q$ to contain only the empty partial assignment (i.e., the assignment that maps every variable to~$\bot$, denoted $[]$
in Algorithm~\ref{alg:dnnf}): note that the $w_C$-score of $[]$ is not
$\bot$ because 
$\rel{C}\neq \emptyset$.  
  We then do the
  following until the queue is empty. We pop (i.e., call \emph{Pop-Max}) from
  the queue a prefix assignment (of maximal~$w_C$-score) that we assign
  to~$\gamma$; we will inductively see that $\gamma$ is a prefix assignment of~$D^{<i}$ for some~$i\in
  \{1,\ldots,n+1\}$ and that its $w_C$-score is not $\bot$. 
  We then do the following for~$j\colonequals i$ to~$n$ (i.e., 
  potentially zero times, in case~$i=n+1$ already). For
  every possible choice of domain element~$d\in D$,
  we let~$\alpha_d$ be the prefix assignment that extends $\gamma$ by
  assigning $x_i$ to~$d$, and we compute the value~$w_C(\alpha_d)$ using
  Lemma~\ref{lem:getwc}. Among these values, the definition of $w_C$ ensures
  that one has a $w_C$-score which is not $\bot$, because this is true
  of~$\gamma$.
  We thus pick a value~$d_0\in D$ such that
  $w_C(\alpha_{d_0})$ is maximal (in particular non-$\bot$). We set~$\gamma$ to~$\alpha_{d_0}$, and we push
  into~$Q$ all other prefix assignments~$\alpha_{d'}$ for~$d' \neq d_0$ for
  which we have~$w_C(\alpha_{d'}) \neq \bot$. 
  Once we have run this for all values of~$j$, we
  have~$i=n+1$, hence~$\gamma$ is a total assignment, and we output it.
  We then continue processing the 
  remaining contents of the queue.

\subparagraph{Correctness of the algorithm.} We show in Appendix~\ref{apx:DNNF} that the
following invariants hold at the beginning and end of every while loop iteration:

\begin{enumerate}
\item\label{it:compatibility} For every~$\tau\in Q$, no satisfying assignment of $C$
  compatible with $\tau$ has been outputted so far;
\item \label{it:incompatibility} For every $\tau, \tau' \in Q$,
  if $\tau \neq \tau'$ then $\tau \not \simeq \tau'$;
\item\label{it:tobeoutputed} For every $\sigma \in \rel C$ that has not yet
  been outputted by the algorithm, there exists some $i\in \{1,\ldots,n+1\}$ such that
    $\sigma|_{X_{<i}} \in Q$ (in fact, the previous point then implies there is at
    most one such~$i$);
\item\label{it:size} The number of elements in~$Q$ is at most~$n \times |D|
  \times (K+1)$, where~$K$ is
  the number of assignments produced so far.
\end{enumerate}
We explain next why they imply
correctness.

\begin{claim}
  \label{alg:dnnfcorrect}
  Algorithm~\ref{alg:dnnf} terminates, enumerates $\rel{C}$ without duplicates
  and in nonincreasing order, and runs with delay $O(|D| \times |C| + \log
  (K+1))$
  with~$K$ the number of assignments produced so far.
\end{claim}

\begin{proof}
  We first show that the algorithm terminates.
Indeed, notice that we pop a prefix assignment from the queue at the beginning of every
while loop iteration. Let us show that, once a prefix assignment~$\tau$ has been
popped
from~$Q$, it cannot be pushed again into~$Q$ for the rest of the algorithm's
execution. 
Indeed, observe that once we pop~$\tau$ from~$Q$, we will first
push to~$Q$ assignments that are strict extensions of~$\tau$ (hence different
from~$\tau$), and then output
a satisfying assignment $\tau'$ of~$C$ that is compatible with~$\tau$, after which the current 
iteration of the while loop ends. Now, by
invariant~(\ref{it:compatibility}), no partial
assignment compatible with~$\tau'$ can ever be added
to~$Q$, and in particular it is the case that $\tau$ cannot ever be added to~$Q$. Thus the queue
becomes empty and the algorithm terminates.

Since the queue eventually becomes empty, by invariant (\ref{it:tobeoutputed}),
the algorithm outputs at least all of~$\rel{C}$. The fact that there are no
duplicates follows from invariant (\ref{it:compatibility}), using a similar
reasoning to how we proved termination. Furthermore, it is clear that only
assignments of~$\rel{C}$ are ever outputted. Therefore the algorithm indeed
enumerates exactly all of $\rel{C}$ with no duplicates.

To check that assignments are enumerated in nonincreasing order, consider an
iteration of the while loop where we
output~$\tau \in \rel{C}$. Let~$\sigma \in \rel{C}$ be an assignment that
has not yet been outputted, and assume by contradiction that~$w(\tau) <
w(\sigma)$. Consider the prefix assignment~$\gamma$ that was popped from the
queue~$Q$ at the beginning of that iteration; clearly by construction we
have~$w_C(\gamma) = w(\tau)$. But by invariant (\ref{it:tobeoutputed}), there
exists a prefix assignment~$\gamma'$ in~$Q$ of which~$\sigma$ is a completion,
hence for this~$\gamma'$ we have~$w_C(\gamma') \geq w(\sigma)$ by definition
of~$w_C$, and this is strictly bigger than $w(\gamma)$, contradicting the fact
that~$\gamma$ had maximal priority.

  Last, we check that the delay between any two consecutive outputs is indeed
  $O(|D|\times |C| + \log (K+1))$. The $O(|D|\times |C|)$ term corresponds to the
  at most~$n \times |D|$ applications of Lemma~\ref{lem:getwc} during a
  for loop until we produce the next satisfying assignment (remember that~$n$ is
  constant so it is not reflected in the delay). The~$O(\log (K+1))$ term corresponds to
  the unique pop operation performed on the priority queue during a while loop
  iteration. Indeed, by invariant~(\ref{it:size}) the queue contains less than~$n
  \times |D| \times (K+1)$ prefix assignments and the complexity of a pop operation is
  logarithmic in this. Since~$n$ is constant we obtain~$O(\log |D| + \log (K+1))$, and
  the~$O(\log |D|)$ gets absorbed in the $O(|D|\times |C|)$ term.
\end{proof}
Thus, up to showing that
the invariants hold (see Appendix~\ref{apx:DNNF}), we have
concluded the proof of Theorem~\ref{thm:dnnf}.

\begin{toappendix}
\begin{claim}
  \label{clm:invar}
  Invariants (\ref{it:compatibility}--\ref{it:size}) hold.
\end{claim}
\begin{proof}
  These invariants
clearly hold at the beginning of the first while loop iteration because then
$Q$ contains only~the empty assignment.

Assuming that the invariants hold at the beginning of some while
loop iteration when~$Q$ is not empty, let us show they still hold at the end of
that iteration. Write~$Q$ the queue at the beginning of that
  iteration,~$\alpha^{i-1} \in D^{<i}$ the prefix assignment that is popped from the
queue and initially assigned to~$\gamma$ and, for~$j\in \{i,\ldots,n\}$, write $\alpha^j$ the
  value of~$\gamma$ after the for loop for this value of~$j$, with~$\alpha^j =
\alpha^{j-1}_{d_j}$ ($= \alpha^{j-1} \times \lit[x_j][d_j]$). Define also~$D'_i
= \{d' \in D \mid d' \neq d_i \text{ and } w_C(\alpha^{i-1}_{d'}) \neq \bot\}$.
Letting~$Q'$ be the queue after this while loop iteration, observe that 
  $Q' = (Q \setminus \{\alpha^{i-1}\}) \cup \bigcup_{i=1}^n \{\alpha^{i-1}_{d'}
\mid d' \in D_i\}$. We check that the invariants hold for~$Q'$.

  We first check invariant (\ref{it:compatibility}): let~$\tau \in Q'$. There are
  two cases. First, if~$\tau
  \in Q \setminus \{\alpha^{i-1}\}$, then, by induction hypothesis using
  invariant~(\ref{it:compatibility}), no satisfying assignment
  of $C$ compatible with~$\tau$ had been outputted
  before~$\alpha^n$. Now, $\tau$ is an extension
  of~$\alpha^{i-1}$ and~$\alpha^{i-1}$ is not
compatible with~$\tau$ by induction hypothesis of invariant
(\ref{it:incompatibility}) on~$Q$. Therefore~$\alpha^n$ cannot be compatible
  with~$\tau$ either. Second, if~$\tau$ is one of the
$\alpha^{i-1}_{d'}$ then by construction it is not compatible with the
  assignment $\alpha^n$ that we
output at this iteration.
  Moreover since~$\tau$ is an extension of~$\alpha^{i-1}$, by induction
  hypothesis with invariant
(\ref{it:incompatibility}) on~$Q$ again we cannot have outputted an
assignment compatible with~$\tau$ before this iteration. Thus, in both cases, (\ref{it:compatibility}) holds for~$Q'$.

Invariant (\ref{it:incompatibility}) is easy to check, again because
  every~$\tau$ of the form $\alpha^{i-1}_{d'}$ is an extension of~$\alpha^{i-1}$
  and~$\alpha^{i-1}$ is incompatible with every~$\tau\in Q \setminus
  \{\alpha^{i-1}\}$ by
induction hypothesis.

  To check invariant (\ref{it:tobeoutputed}), consider
  $\sigma \in \rel{C}$. By induction hypothesis with the same invariant, there
  is some $i \in \{1, \ldots, n+1\}$ such that $\sigma|_{X_{<i}} \in Q$. 
  So if $\sigma|_{X_{<i}}$ is not~$\alpha^{i-1}$ then the claim is immediate.
  Otherwise, we can see that we either produce~$\sigma$ or add to~$Q'$ some
  assignment that extends $\alpha^{i-1}$ and is compatible with~$\sigma$.

  Finally, invariant (\ref{it:size}) is also trivial to show by applying
  induction hypothesis to~$Q$ and observing that $\bigcup_{i=1}^n
  \{\alpha^{i-1}_{d'} \mid d' \in D_i\}$ has at most~$n \times |D|$ elements.
\end{proof}
  \end{toappendix}

\section{Ranked Enumeration for Smooth Multivalued d-DNNFs}
\label{sec:d-DNNF}
Having shown our polynomial-delay ranked enumeration algorithm for DNNF
circuits, we move on in this section to our main technical contribution.
Specifically, we present an algorithm for smooth multivalued DNNF circuits that
are further assumed to be \emph{deterministic}, 
but which achieves linear-time preprocessing and delay $O(\log (K+1))$, where 
$K$ denotes the number of satisfying assignments produced so far. This proves
Result~\ref{res:ddnnf}, which we restate below:

\begin{theorem}
  \label{thm:ddnnf}
  For any constant~$n\in \NN$, we can solve the $\Enum$ problem on an
  input smooth
  multivalued d-DNNF circuit~$C$ with~$n$ variables 
  and a subset-monotone ranking function, with preprocessing $O(|C|)$ and
  delay \mbox{$O(\log (K+1))$}, where $K$ is
  the number of
  assignments produced so far.
\end{theorem}

Let us fix for this section the set $X$ of variables of~$C$ (with~$|X|=n$)
and the domain~$D$. 

The rest of this section is devoted to proving Theorem~\ref{thm:ddnnf}. It is
structured in three subsections, corresponding to the three main technical
difficulties to overcome. First, we explain in Section~\ref{subsec:preproc} the
preprocessing phase of the algorithm, where in particular we use Brodal queues
to quickly “jump” over~$\uplus$-gates. Second, in Section~\ref{subsec:AodotB}, we present a
simple algorithm, that we call the \emph{$A\odot B$ ranked enumeration
algorithm}, which conveys in a self-contained fashion the idea of how 
we handle~$\times$-gate
during the enumeration phase of the main algorithm. Last, we present the
enumeration phase in Section~\ref{subsec:enum}.

\subsection{Preprocessing Phase}
\label{subsec:preproc}

The preprocessing phase is itself subdivided in four steps, described next.
\subparagraph*{Preprocessing: first step.}
We preprocess~$C$ in~$O(|C|)$ to ensure that the $\times$-gates of the circuit always have
exactly two inputs. This can easily be done as follows. Remember that
our definition of multivalued circuits does not allow $\times$-gates with no
inputs, so this case does not occur.
We can then eliminate $\times$-gates with one input by
replacing them by their single input. Next, we can rewrite $\times$-gates with
more than two inputs to replace them by a tree of $\times$-gates with two
inputs. For simplicity, let us call~$C$ again the resulting smooth multivalued
d-DNNF circuit in which $\times$-gates always have exactly two inputs.

\subparagraph*{Preprocessing: second step.}
We compute, for every gate~$g$ of $C$ the value~$\#g \colonequals |\rel{g}|$.
This can clearly be done in linear time again, by a bottom-up traversal of~$C$
and using decomposability, determinism and smoothness. Note that~$\#g$ has value at most~$|D|^n$, which is
polynomial (as $n$ is a constant), so this fits into one memory cell.

\subparagraph*{Preprocessing: third step.}
The third step begins by
initializing for every gate~$g$
of~$C$ an empty Brodal queue~$B_g$.
We then populate those queues by a
(linear-time) bottom-up traversal of the circuit, described next.
This traversal will add
to each queue~$B_g$ some priority-data pairs of the form~$(\pp: w(\tau),~ \dd:
(g',1,\tau))$ where~$g'$ has a (possibly empty) directed path to~$g$ and~$\tau
\in \rel{g}$.
We will shortly explain what is the exact content of these queues at the end of
this third preprocessing step, but we already point out one invariant: once we
are done processing a gate $g$ in the traversal, then $B_g$ contains at least
one priority-data pair of this form, i.e., it is non-empty.

The traversal proceeds as follows:
\begin{itemize}
  \item If~$g$ is an input gate labeled with
$\langle x : d \rangle$ corresponding to
the singleton assignment~$\alpha = [x \mapsto d]$, then we push into~$B_g$ the
priority-data pair corresponding to this assignment: $(\pp: w(\alpha),~ \dd: (g,1,\alpha))$.
    \item If~$g$ is a~$\times$-gate with
      inputs~$g_1$ and~$g_2$ then we call \emph{Find-Max} on the
      Brodal queues~$B_{g_1}$ and~$B_{g_2}$ of the inputs. These gates $g_1$ and
      $g_2$ have already been processed, so the queues $B_{g_1}$ and $B_{g_2}$
      are non-empty, and we 
obtain priority-data pairs $(\pp: w(\tau_1),~ \dd: (g'_1,1,\tau_1))$ and
$(\pp: w(\tau_2),~ \dd: (g'_2,1,\tau_2))$, where~$\tau_1\in \rel{g_1}$ and
$\tau_2\in \rel{g_2}$. We push into~$B_g$ the pair $(\pp: w(\tau_1\times
\tau_2),~ \dd: (g,1,\tau_1 \times \tau_2))$.
    \item If~$g$ is a~$\uplus$-gate with
input gates~$g_1,\ldots,g_m$ then we set~$B_g$ to be the union
of~$B_{g_1},\ldots,B_{g_m}$; recall that the union operation on two Brodal
queues can be done in~$O(1)$, so that this union is linear in~$m$.
\end{itemize}
It is clear that this third preprocessing step takes time $O(|C|)$.
To describe what the queues contain at the end of this step,
we need to define the notion of \emph{exit gate} of a~$\uplus$-gate:

\begin{definition}
  \label{def:exit}
For a $\uplus$-gate~$g$ of~$C$, an \emph{exit gate of~$g$} is a gate $g'$ which
is not a $\uplus$-gate (i.e., a $\times$-gate or an input of the circuit) such
that there is a path from $g'$ to $g$ where every gate except $g'$ on this path is
a $\uplus$-gate. We denote by $\exit{g}$ the set of exit gates for $g$. 
\end{definition}

We can then characterize what the queues contain:
\begin{claim}
  \label{claim:preproc-step-3}
When the third preprocessing step finishes, the queues~$B_g$ are as follows.
\begin{itemize}
  \item If~$g$ is an input gate corresponding to the
    singleton assignment~$\alpha = [x \mapsto d]$ then~$B_g$ contains only the
    pair $(\pp: w(\alpha),~ \dd: (g,1,\alpha))$.
  \item If~$g$ is a~$\times$-gate then~$B_g$ contains only one pair, which is of the form $(\pp:
    w(\tau),~ \dd: (g,1,\tau))$ where~$\tau$ is some satisfying assignment of~$g$
    of maximal score  among $\rel{g}$.
  \item If~$g$ is a~$\uplus$-gate then~$B_g$ contains exactly the following:
    for every exit gate~$g'$
    of~$g$, the queue $B_g$ contains one pair of the form $(\pp: w(\tau),~ \dd: (g',1,\tau))$ where~$\tau$
    is some satisfying assignment of~$g'$ of maximal score  among $\rel{g'}$. 
\end{itemize}
\end{claim}
This implies, in particular, that for every~$g\in C$ the queue~$B_g$
contains a pair $(\pp: w(\tau),~ \dd: (g',1,\tau))$ (possibly with~$g'=g$)
where~$\tau$ is a satisfying assignment of~$g$ of maximal score among $\rel{g}$.
\begin{proof}[Proof of Claim~\ref{claim:preproc-step-3}]
  It is routine to prove this by bottom-up induction, in particular using
  Lemma~\ref{lemma:maxdnnf} for the case of~$\times$-gates.
\end{proof}
This concludes the third preprocessing step. Intuitively, the Brodal
queues computed at this step will allow us to jump directly to the exits
of~$\uplus$-gates, without spending time traversing potentially
long paths of~$\uplus$-gates.
Thanks to the constant-time 
union operation on Brodal queues, this third step takes linear time, and in
fact
this is the only part of the proof where we need this bound on the union
operation. More precisely, in the
remainder of the algorithm, we will only use on priority queues $Q$ the operations
\emph{Initialize}, \emph{Push} and \emph{Find-Max} (in~$O(1)$) and
\emph{Pop-Max} (in~$O(\log |Q|)$).

\subparagraph*{Preprocessing: fourth step.}
In the fourth and last preprocessing step, we define some more data structures on
every gate~$g$ of~$C$. 

First, we define for every gate~$g$ a
priority queue~$Q_g$.
For all input gates and $\uplus$-gates, we simply set $Q_g \colonequals
B_g$, but for $\times$-gates we will define $Q_g$ to be new priority queues.
Once this is done, we will only use the priority queues $Q_g$, and can forget about the priority queues~$B_g$.
We construct $Q_g$ for each $\times$-gate $g$ separately, in $O(1)$ time, as
follows.
Letting $g_1$ and $g_2$ be the inputs to~$g$, we call Find-Max
    on~$B_g$. By Claim~\ref{claim:preproc-step-3}, we obtain a
    pair $(\pp: w(\tau),~ \dd: (g,1,\tau))$ where~$\tau$ is some satisfying
    assignment of~$g$ of maximal score among $\rel{g}$. We split $\tau$ into~$\tau_1 \times \tau_2$
    where~$\tau_i \in \rel{g_i}$ for $i \in \{1,2\}$,
    and we define the priority queue~$Q_g$ to contain one priority-data pair,
    namely, $(\pp: w(\tau),~ \dd: (1,1,\tau_1,\tau_2))$.

    Second, we allocate for every gate $g$ a table~$T_g$ of size~$\#g$ (indexed
    starting from~$1$), that will later hold satisfying assignments of~$g$ in
    nonincreasing order of scores, stored into contiguous memory cells starting at the
    beginning of~$T_g$. We do not bother initializing these tables, but we initialize
    integers~$i_g$ to~$0$, that will store the current number of assignments
    stored in~$T_g$. 

    Last, we also initialize to~$0$ a bidimensional bit table~$R_g$ for
    every~$\times$-gate~$g$, of size~$\#g_1 \times \#g_2$ with~$g_1,g_2$ the
    two inputs of~$g$. This can be done in~$O(1)$ with the technique of \emph{lazy
    initialization}, see e.g.,~\cite[Section 2.5]{grandjean2022arithmetic}. The role of
    these tables will be explained later.

This concludes the description of the preprocessing phase of our algorithm. 
In what follows, we will rely on the priority queues $Q_g$, the tables $T_g$,
the integers $i_g$ storing their size, and the tables~$R_g$,
The following should then be clear:
\begin{claim}
  \label{claim:preproc-step-4}
  Once we finish the fourth preprocessing step (concluding the
  preprocessing), all integers~$i_g$ are~$0$, all tables $T_g$ and~$R_g$ are empty, and the
queues~$Q_g$ contain the following.
\begin{itemize}
  \item If~$g$ is an input gate corresponding to the
    singleton assignment~$\alpha = [x \mapsto d]$, then~$Q_g$ contains only the
    pair $(\pp: w(\alpha),~ \dd: (g,1,\alpha))$.
  \item If~$g$ is a~$\times$-gate with inputs~$g_1,g_2$, then~$Q_g$ contains 
only one priority-data pair which is of the form~$(\pp: w(\tau_1\times \tau_2),~ \dd: (1,1,\tau_1,\tau_2))$,
    where~$\tau_1\times \tau_2$ is some satisfying assignment of~$g$
    of maximal score among $\rel{g}$.
  \item If~$g$ is a~$\uplus$-gate, then~$Q_g$ contains, for every exit gate~$g'$
    of~$g$, one pair of the form $(\pp: w(\tau),~ \dd: (g',1,\tau))$ where~$\tau$
    is some satisfying assignment of~$g'$ of maximal score among $\rel{g'}$. 
\end{itemize}
\end{claim}
Again, this in particular implies that each~$Q_g$ stores a satisfying
assignment of~$g$ of maximal score (but the way in which it is represented
depends on the type of~$g$).

\subsection{$A\odot B$ Ranked Enumeration Algorithm}
\label{subsec:AodotB}

Having described the preprocessing phase,
we present in this section a component of the enumeration phase of our
algorithm,
called the \emph{$A \odot B$ ranked enumeration} algorithm.
This simple algorithm will be used at every~$\times$-gate $g$
during the enumeration phase to enumerate all ways to combine the assignments of
the two inputs of~$g$.

Let~$\odot:\RR\times \RR \to \RR$ be an operation which is computable
in $O(1)$ and such that, for
all~$a\leq a'$ and~$b\leq b'$ we have~$a \odot b \leq a' \odot b'$ (this
is similar to subset-monotonicity, and is in fact equivalent; cf.\ 
Lemma~\ref{lem:subsetmonotone} in the appendix).
We explain in this section how, given as input
two tables (indexed starting from~$1$)~$A,B$ of reals of size~$n_1,n_2$
sorted in nonincreasing order, we can enumerate the set of integer
pairs~$\{(i,j) \mid (i,j) \in
\{1,\ldots,n_1\}\times \{1,\ldots,n_2\}\}$, 
in nonincreasing order of the
score~$A[i] \odot B[j]$, with~$O(1)$ preprocessing and a delay~$O(\log K)$
where~$K$ is the number of pairs outputted so far. 

Intuitively, this will be applied at every~$\times$-gate~$g$, with~$[n_1]$
(resp.,~$[n_2]$) representing the satisfying valuations of the first (resp.,
second) input of~$g$ sorted in a nonincreasing order, and $A[l]$ (resp., $B[l]$) representing
the score of those valuations according to the ranking function.

The algorithm is shown in
Algorithm~\ref{alg:toy}, but we also paraphrase it in text with more
explanations. We initialize a two-dimensional bit table $R$ of size~$n_1 \times
n_2$ to contain only zeroes (again using lazy
initialization~\cite[Section 2.5]{grandjean2022arithmetic}), whose
role will be to remember which pairs have been seen so far, and a priority
queue~$Q$ containing only the pair~$(\pp: A[1] \odot B[1] ,~ \dd: (1,1))$; we
set $R[1,1]$ to true because the pair $(1, 1)$ has been seen. Then,
while the queue is not empty, we do the following. We pop (call
\emph{Pop-Max}) from~$Q$, obtaining a priority-data pair of the form $(\pp:
A[i] \odot B[j],~ \dd: (i,j))$. We output the pair~$(i,j)$. Then, for each
$(p,q) \in \{(i+1,j),(i,j+1)\}$ that is 
in the~$[n_1]\times[n_2]$ grid, if the pair $(p,q)$ has not been seen before, then we push
into~$Q$ the pair~$(\pp:A[p] \odot B[q],~ \dd:
(p,q))$ and mark $(p,q)$ as seen in~$R$. We show the following, with a full proof in Appendix~\ref{apx:toycorrect}.

\begin{algorithm}[t]
\SetAlgoLined
\KwData{Two arrays~$A,B$ of real numbers of size~$n_1,n_2$ (indexed from~$1$),
  sorted in nonincreasing order; An operation~$\odot$ as described in the main
  text.}
\KwResult{An enumeration of the pairs~$\{(i,j) \mid (i,j) \in
\{1,\ldots,n_1\}\times \{1,\ldots,n_2\}\}$ in nonincreasing order of the
score~$A[i] \odot B[j]$}
  $R \gets$ bidimensional array of size~$n_1\times n_2$ lazily initialized
  to~$0$\;
  $Q \gets \text{empty priority queue}$\;
  Push $(1,1)$ into $Q$ with priority~$A[1] \odot B[1]$\;
  $R[1,1] \gets \text{true}$\;
\While{$Q$ is not empty}{
  \text{Pop into} $(i,j) \text{~the pair with maximal priority from~} Q$\;
  Output $(i,j)$\;
  \For{$(p,q) \in \{(i+1, j), (i, j+1)\}$} {
    \If{$p\leq n_1$ and $q\leq n_2$ and $R[p][q] = 0$}{
      Push $(p,q)$ into $Q$ with priority $A[p] \odot B[q]$\;
      $R[p][q] \gets \text{true}$\;
    }
  }
}
  \caption{Algorithm for $A\odot B$ ranked enumeration}
  \label{alg:toy}
\end{algorithm}

\begin{toappendix}
  \subsection{Proof of Claim~\ref{clm:toycorrect}}
  \label{apx:toycorrect}
\end{toappendix}
\begin{claimrep}
  \label{clm:toycorrect}
  This $A\odot B$ ranked enumeration algorithm is correct and runs with the
  stated complexity.
\end{claimrep}
\begin{proofsketch}
  The proof is simple and hinges on the following two invariants: 
  \begin{enumerate}
    \item For any pair~$(i,j)$ not enumerated so far, there
      exists a pair~$(i',j')$ (possibly~$(i,j)=(i',j')$) such that~$(i',j')$ is
      in~$Q$,
      and a simple path in the~$[n_1]\times[n_2]$ grid from~$(i',j')$ to~$(i,j)$
      with nondecreasing first and second coordinates such that none of the
      pairs in that path have been outputted yet. 
    \item The queue contains at most~$K+1$ pairs for~$K$ 
      the number of pairs outputted so far.\qedhere
  \end{enumerate}
\end{proofsketch}

\begin{proof}
  We claim that the following invariants hold at the beginning and end of every
  while loop iteration. To simplify notation, when we talk of the
  \emph{$\odot$-score of a pair~$(i,j)\in [n_1]\times[n_2]$} we mean the
  score~$A[i]\odot B[j]$.
  \begin{enumerate}
    \item\label{it:three} For any pair~$(i,j)$ not enumerated so far, there
      exists a pair~$(i',j')$ (possibly~$(i,j)=(i',j')$) such that~$(i',j')$ is
      in~$Q$,
      and a simple path in the~$[n_1]\times[n_2]$ grid from~$(i',j')$ to~$(i,j)$
      with nondecreasing first and second
      coordinates\footnote{More precisely, if we see~$(1,1)$ as being at the bottom left of the grid, then the path
      always goes either up or right.} such that none of the
      pairs in that path have been outputted yet. 
    \item\label{it:four} The queue contains at most~$K+1$ pairs where~$K$ is
      the number of pairs outputted so far.
  \end{enumerate}
  Before proving that these invariants hold, let us argue that they imply the
  claim. It is clear
  that the algorithm terminates thanks to the array~$R$: once a pair is added to
  the queue its~$R$-value is set to true (and never set to false again), so
  it cannot be put back into~$Q$ anymore, and so the queue eventually becomes
  empty.

  We then show that the algorithm enumerates all of~$[n_1]\times[n_2]$:
  indeed, assume by contradiction that the algorithm terminates and that some
  pair~$(i,j)$ has not been outputted. The queue is empty, but invariant
  (\ref{it:three}) implies that there is still an element~$(i',j')$ in the
  queue, a contradiction. It is also clear by construction that the algorithm
  only enumerates pairs of $[n_1] \times [n_2]$. It is clear that it does not
  enumerate the same pair twice thanks to the table~$R$: we only add each pair
  at most once to the queue, hence it will be popped and enumerated at most once.
  The claim for the delay is also clear thanks to invariant (\ref{it:four}).

  Last, we show that the pairs are enumerated in nonincreasing order
  of~$\odot$-score.
  Assume by way of contradiction that
  pair~$(i_1,j_1)$ is enumerated before pair~$(i_2,j_2)$ and that $A[i_1] \odot
  B[j_1] < A[i_2] \odot B[j_2]$. Now, consider the beginning of the while-loop
  iteration where we are about to pop~$(i_1,j_1)$ from~$Q$. Since~$(i_2,j_2)$
  has not been enumerated yet, by invariant (\ref{it:three}) there
  exists~$(i'',j'')$ in~$Q$ with~$\odot$-score at least that of~$(i_2,j_2)$ (by monotonicity of~$\odot$ along the path);
  this contradicts the fact that~$(i_1,j_1)$ has maximal priority when we pop
  it from~$Q$.

  We now prove the invariants. They are true before the first while
  loop iteration: indeed, nothing has been outputted yet so
  invariant~(\ref{it:four}) holds. Further,
  the queue only contains the pair~$(1,1)$, and the subset-monotonicity-like
  property that we assumed on the
  function $\odot$ implies that, for any pair $(i,j)$ of the grid
  $[n_1]\times[n_2]$, the path $(1,1), \ldots, (1,j), (2,j), \ldots, (i,j)$ is a
  nondecreasing path from~$(1,1)$ to~$(i,j)$ satisfying the invariant.

  Assume that the invariants are true at the beginning of some while
  loop iteration and the queue is not empty, and let us show that they still hold
  at the end of that iteration. Call~$Q$ the state of the queue before that
  iteration, $(i,j)$ the pair that is popped, and~$Q'$ the state of the queue
  afterwards. Invariant (\ref{it:four}) is clear by induction hypothesis: we
  popped a pair and added at most two during that iteration.

  Now for
  invariant~(\ref{it:three}), assume by contradiction that there exists a
  pair~$(i_1,j_1)$ which has not been enumerated after the end of that
  iteration and which does not satisfy the invariant. 
  In particular, we have not enumerated~$(i_1,j_1)$ before that iteration, so
  applying the induction hypothesis we obtain a pair $(i',j')$
  (possibly~$(i_1,j_1)=(i',j')$) such that~$(i',j')$ is in~$Q$ and a simple
  path in the~$[n_1]\times[n_2]$ grid from~$(i',j')$ to~$(i_1,j_1)$ with
  nondecreasing first and second coordinates such that none of the pairs in
  that path have been outputted before that iteration. If~$(i,j)$ is not one of
  the pairs of that path then we are done because this path still satisfies the
  conditions at the end of the iteration. Otherwise, consider the pair~$(u,v)$
  after~$(i,j)$ in that path.
  Then~$(u,v)$ is either to the right of~$(i,j)$ or above~$(i,j)$ because the
  path has nondecreasing first and second coordinates.
  Hence, either $(u,v)$ is already in~$Q$ (and therefore will also be in~$Q'$), or
  it is not in~$Q$ and has been pushed into~$Q'$ when
  extracting~$(i,j)$.
  So the path starting at~$(u,v)$ is a valid witnessing path.

  This establishes that the invariants hold and
  concludes the proof of Claim~\ref{clm:toycorrect}.
\end{proof}
\begin{toappendix}
  \subsection{Correctness of the Enumeration Phase (Section~\ref{subsec:enum})}
\end{toappendix}
\subsection{Enumeration Phase}
\label{subsec:enum}

We last move on to the enumeration phase. 
We first give a high-level description of how the enumeration phase works,
before presenting the details.

\subparagraph*{The operation~$\get(g,j)$.}
We will define a recursive operation~$\get$, running in
complexity~\mbox{$O(\log (K+1))$}, that applies to a gate~$g$ 
and integer~$1\leq j \leq i_g + 1$ and does the following. If~$j
\leq i_g$ then $\get(g,j)$ simply returns the satisfying assignment of~$g$ that
is stored in~$T_g[j]$ (i.e., this assignment has already been computed).
Otherwise, if $j=i_g+1$, then $\get(g,j)$ finds the next assignment to be
enumerated, inserts it into~$T_g$, and returns
that assignment. Note that, in this case, calling $\get(g,j)$ modifies the
memory for~$g$ and some other gates~$g'$. Specifically, it modifies 
the tables~$T_{g'}$ and~$R_{g'}$, the queues~$Q_{g'}$, and the integers~$i_{g'}$
for various gates~$g'$ having a directed path to~$g$ (i.e., including
$g'=g$).

When we are not executing an operation~$\get$,
the memory will satisfy the following invariants, for every~$g$ of~$C$:
\begin{itemize}
  \item The table~$T_g$ contains 
    assignments~$\tau\in \rel{g}$, ordered by nonincreasing
    score and with no duplicates; and $i_g$ is the current size of~$T_g$;
  \item For any assignment~$\tau\in \rel{g}$ that does not occur in~$T_g$, it
    is no larger than the last assignment in~$T_g$, i.e., we
    have~$w(\tau) \leq w(T_g[i_g])$.
  \item The queues $Q_g$ will also satisfy some invariants, which will be
    presented later. 
  \item The tables~$R_g$ for the~$\times$-gates record whether we have already
    seen pairs of satisfying assignments of the two children, similarly to how
    this is done in the~$A\odot B$ algorithm.
\end{itemize}
The tables $T_g$ store the assignments in the order in which we find 
them, which is compatible with the ranking function. This allows us, in particular,
to obtain in constant time the~$j$-th satisfying assignment of~$\rel{g}$
if it has already been computed, i.e., if~$j
\leq i_g$. The reason why we keep the assignments in the tables $T_g$ is because
we may reach the gate~$g$ via many different paths throughout the enumeration,
and these paths may be at many different stages of the enumeration on~$g$.

At the top level, if we can implement $\get$ while satisfying the invariants
above, then the enumeration phase of the algorithm is simple to describe:
for~$j$ ranging from~$1$ to~$\#r$, we output~$\get(r,j)$, where~$r$ is the output gate
of~$C$.

\subparagraph*{Implementing~$\get$.}
We first explain the intended semantics of data
values in the queues~$Q_b$:
\begin{itemize}
  \item If~$g$ is a~$\uplus$-gate then~$Q_b$ will always
contain pairs of the form $(\pp: w(\tau),~ \dd: (g',j,\tau))$ where~$g' \in
\exit{g}$ and~$j \in \{1,\ldots,i_{g'}+1\}$ and~$\tau \in \rel{g'}$, and the idea is that at the end of
the enumeration~$\tau$ will be stored at position~$j$ in~$T_{g'}$.
    \item If~$g$ is
a~$\times$-gate, letting $g_1'$ and $g_2'$ be the input gates, then~$Q_b$ will always contain pairs of the form $(\pp:
w(\tau_1\times \tau_2),~ \dd: (j_1,j_2,\tau_1,\tau_2))$ with~$\tau_i \in
\rel{g_i}$ and at the end of the enumeration~$\tau_i$ will be
at position~$j_i$ in~$T_{g_i}$ with~$j_i \in \{1,\ldots,i_{g'_i}+1\}$ for
    all $i \in \{1, 2\}$. 
\item If $g$ is an input gate, then $Q_b$ initially contains the only
  assignment captured by~$g$, becomes empty the first time we call $\get(g,1)$,
  and remains empty thereafter.
\end{itemize}
The implementation of~$\get$ is given in Algorithm~\ref{alg:get}. Intuitively,
the algorithm for~$\uplus$-gates simply consists of interleaving the maximal
assignments of its exit gates, similarly to how one builds a sorted list for
the union of two or more sorted lists. Here, determinism ensures that we do not
get duplicates. The algorithm for~$\times$-gates proceeds similarly to
the~$A\odot B$ algorithm, as explained in the previous section.

This concludes the presentation of the function $\get$, and with it that of
the enumeration phase of the algorithm. The discussion of the delay bound is deferred to Appendix~\ref{apx:correct}.

\begin{algorithm}
\SetAlgoLined
\KwData{The tables~$T_g,R_g$, queues~$Q_g$, integers~$\#g, i_g$, ranking
function~$w$, a gate~$g$, and integer~$j\in
\{1,\ldots i_g +1\}$.}
\KwResult{The~$j$-th satisfying assignment of~$g$.}
 \lIf{$j \leq i_g$}{
   \Return $T_g[j]$
    }
    \tcp{From now on, we have~$j=i_g + 1$}
 \If{$g$ is an input gate}{
   $(\pp: \delta,~ \dd: (g,1,\tau')) \gets $ \text{Pop from }$Q_j$\;
   $\tau \gets \tau'$\;
 }
 \ElseIf{$g$ is a $\uplus$-gate}{
 $(\pp: \delta,~ \dd: (g',j',\tau')) \gets $ \text{Pop from }$Q_j$\;
 $\tau \gets \tau'$\;
 \If{$j'+1 \leq \#{g'}$}{
     $\tau'' \gets \get(g',j'+1)$\;
     \text{Push into} $Q_g$ \text{the priority-data pair} $(\pp: w(\tau''),~ \dd: (g',j'+1,\tau''))$\;
 }
 }
 \ElseIf{$g$ is a $\times$ gate}{
 $(\pp: \delta,~ \dd: (j_1,j_2,\tau_1,\tau_2)) \gets $ \text{Pop from }$Q_j$\;
 $\tau \gets \tau_1 \times \tau_2$\;
  \For{$(p,q) \in \{(j_1+1, j_2), (j_1, j_2+1)\}$} {
    \If{$p \leq \#{g_1}$ and $q \leq \#{g_2}$ and $R_g[p][q] = \text{false}$}{
      $\tau'_1 \gets \get(g_1,p)$\;
      $\tau'_2 \gets  \get(g_2,q)$\;
      $\tau' \gets \tau'_1 \times \tau'_2$\;
     \text{Push into} $Q_g$ \text{the priority-data pair} $(\pp: w(\tau'),~ \dd: (p,q,\tau'))$\;
      $R_g[p][q] \gets \text{true}$\;
    }
   }
  }
  $T_g[i_g+1] \gets \tau$\;
  $i_g \gets i_g + 1$\;
   \Return $\tau$
  \caption{Implementation of~$\get(g,j)$ for the enumeration phase}
  \label{alg:get}
\end{algorithm}

\begin{toappendix}
  \label{apx:correct}

  The fact that algorithm correctly enumerates~$\rel{C}$ in nonincreasing order
  of scores
  and without duplicates follows from the invariants and observations that we have already
  hinted at during the presentation of the algorithm (we do not give the full
  details of the proof). Thus, what remains to be shown is the~$O(\log K)$ bound
  on the delay.

Observe that the delay is simply the total time it takes for a call of the
operation $\get$ on the output gate~$r$ of~$C$ to finish; call such a call a
\emph{top-level call}. Now, during a top-level call, notice that for every
gate~$g'$, the queue~$Q_{g'}$ grows by at most one element, just like in
the~$A\oplus B$ algorithm: this is because, thanks to decomposability, the “trace”
of the algorithm is a tree, hence~$\get$ is called at most once
on~$g'$ during any top-level call. Therefore, after~$K$ assignments have been
outputted, each queue has size at most~$K+1$.

Moreover thanks to the way we “jump”~$\uplus$-gates by directly going to their
exit nodes and thanks to decomposability, 
the total number of recursive calls to $\get$
caused by a top-level call is linear in the number~$n$ of variables, which
is constant. The bound follows, since during a call of~$\get(g)$ for some
gate~$g$, not counting the time for the recursive calls of~$\get$ on
descendants of~$g$, we do a constant amount of computation plus exactly one
call to \emph{Pop-Max} on~$Q_g$, which is of logarithmic complexity.
\end{toappendix}

\section{Application to Monadic Second-Order Queries}
\label{sec:mso}
Having presented our results on ranked enumeration for smooth multivalued DNNFs and d-DNNFs, we 
present their consequences in this section for the problem of ranked enumeration
of MSO query answers on trees. We first present some preliminaries on trees and
MSO, formally define the evaluation problem, and explain how to reduce it to our
results on circuits.

\subparagraph*{Trees and MSO on trees.}
We fix a finite set $\Lambda$ of \emph{tree labels}.
A \emph{$\Lambda$-tree} is then a tree~$T$ whose nodes carry a label from~$\Lambda$,
and which is rooted, ordered, binary, and full, i.e., every node has either no
children (a \emph{leaf}) or exactly one \emph{left child} and one \emph{right
child} (an \emph{internal node}). We often abuse notation and write $T$ to refer
to its set of nodes.

We consider \emph{monadic second-order logic} (MSO) on trees, which 
extends first-order logic with quantification
over sets. The signature of MSO on $\Lambda$-trees allows us
to refer to the left child and right child relationships along with
unary predicates referring to the node labels; and it can express,
e.g., the set of descendants of a node. We only consider MSO queries where the
free variables are first-order. We omit the precise semantics of MSO; see, e.g.,~\cite{Libkin04}.

Fixing an MSO query $\Phi(x_1, \ldots, x_n)$ on $\Lambda$-trees, given a
$\Lambda$-tree $T$, the \emph{answers} of $\Phi$ on~$T$ are the assignments
$\alpha$ on
variables $X = \{x_1, \ldots, x_n\}$ and domain~$T$ such that $\Phi(\alpha)$ holds
on~$T$ in the usual sense. It is known that, for any such query~$\Phi$, given~$T$
and an assignment~$\alpha$ from~$X$ to~$T$, 
we can check whether $\Phi(\alpha(X))$
holds in linear time. What is more, given~$T$, we can enumerate the answers
of~$\Phi$ on~$T$ with linear preprocessing and constant
delay~\cite{bagan2006mso,kazana2013enumeration,amarilli2017circuit}.

We now define \emph{ranked enumeration}. For a tree $T$ and
variables $X = \{x_1, \ldots, x_n\}$, a $(T,X)$-ranking function is simply a
ranking function as in Section~\ref{sec:preliminaries}, whose domain
is the set of nodes of~$T$. We still assume that ranking functions are
subset-monotone. The \emph{ranked enumeration} problem for a
fixed MSO query $\Phi$ with variables~$X$, also denoted $\Enum$, takes an input a
tree~$T$ and a subset-monotone $(T,X)$-ranking function~$w$, and must enumerate
all answers of~$\Phi$ on~$T$, without duplicates, in nonincreasing order of
scores (with ties broken arbitrarily).

\subparagraph*{Ranked enumeration for MSO.}
We are now ready to restate Result~\ref{res:mso} from the introduction:

\begin{theorem}
  \label{thm:mso}
  For any fixed tree signature $\Lambda$ and MSO query $\Phi$ on variables $X$
  on~$\Lambda$-trees,
  given a $\Lambda$-tree $T$ and a subset-monotone $(T,X)$-ranking function~$w$,
  we can solve the $\Enum$ problem for~$\Phi$ on~$T$ and~$w$ with preprocessing
  time $O(|T|)$ and delay $O(\log (K+1))$ where $K$ is the number of answers
  produced so far.
\end{theorem}

Recall that, as the total number of answers is at most $|T|^{|X|}$ and $|X|$ is
constant, then this implies a delay bound of $O(\log |T|)$.
The result is simply shown 
by constructing a smooth multivalued d-DNNF representing the query answers.
This
can be done in linear time with existing
techniques (we provide a self-contained proof in
Appendix~\ref{apx:mso}):

\begin{proposition}[\cite{amarilli2017circuit,amarilli2019enumeration}]
  \label{prp:compil}
  For any fixed tree signature $\Lambda$ and MSO query $\Phi$ on variables $X$
  on~$\Lambda$-trees,
  given a $\Lambda$-tree $T$, we can check in time $O(|T|)$ if $\Phi$ has some
  answers on~$T$, and if yes
  we can build in time $O(|T|)$ a smooth multivalued d-DNNF $C$ on domain~$T$
  and variables~$X$ such that $\rel{C}$ is precisely the set of answers
  of~$\Phi$
  on~$T$.
\end{proposition}

Note that we exclude the case where $\Phi$ has no answer on~$T$, because our definition of
multivalued circuits does not allow them to capture an empty set of assignments;
of course we can do this check in the preprocessing, and if there are no answers
then enumeration is trivial.

These results are intuitively shown by translating the MSO query to a tree
automaton, and then computing a provenance circuit of this
automaton by a kind of product construction~\cite{amarilli2015provenance}. The
resulting circuit is a smooth multivalued DNNF, and is additionally a d-DNNF if
the automaton is deterministic. We can then show Theorem~\ref{thm:mso} simply by
performing the compilation (Proposition~\ref{prp:compil}) as part of the
preprocessing, and then invoking the enumeration algorithm of
Section~\ref{sec:d-DNNF} (Theorem~\ref{thm:ddnnf}). Notice that we could also
use the algorithm of Section~\ref{sec:DNNF} (Theorem~\ref{thm:dnnf}), in
particular if it is easier to obtain a nondeterministic tree automaton for the
query, as its provenance circuit would then be a non-deterministic
DNNF~\cite{amarilli2019enumeration}.

\begin{toappendix}
  \section{Self-Contained Proof of Proposition~\ref{prp:compil}}
\label{apx:mso}

In this appendix, we now explain how to prove Proposition~\ref{prp:compil} with
the circuit definitions of the present paper. We re-state that this
self-contained proof is provided only for convenience, and that it features no
conceptual contributions relative to existing
results~\cite{amarilli2017circuit,amarilli2019enumeration}.

\subparagraph*{Changing the MSO query.}
Given the tree alphabet $\Lambda$ used in the MSO query, and letting $X$ be the
set of variables, we work with an extended alphabet $\Gamma \colonequals \Lambda
\times 2^X$. We say that a $\Gamma$-tree $T$ is \emph{well-formed} if, for each $x
\in X$, there is precisely one tree node $n$ of~$T$ such that, writing
$(\lambda, s)$ the label of~$n$, we have $x \in s$.
For a $\Lambda$-tree $T$, given an assignment~$\alpha$ on domain $T$ and
variables $X$, we write $T_\alpha$ the well-formed $\Gamma$-tree where we label each node
with the subset of variables assigned to that node, formally, $T_\alpha$ has the
same skeleton as $T$ and the label of a node $n$ in~$T_\alpha$ is $(\lambda, s)$
where $\lambda \in \Lambda$ is the label of $n$ in~$T$ and $s$ is the subset
of~$X$ such that $\alpha(x) = n$ for all~$x\in s$. For a MSO query $\Phi$ on $\Lambda$-trees with
variables $X$, we denote by~$\Phi'$ the Boolean query on $\Gamma$-trees which
accepts precisely the $\Gamma$-trees $T'$ that are well-formed and where, letting $T$ be the
$\Lambda$-tree obtained from~$T'$ by projecting the labels on their first
component, and letting $\alpha$ be the unique assignment such that $T' =
T_\alpha$, then $\Phi(\alpha)$ holds on~$T$. We claim:

\begin{claim}
  \label{clm:phi2}
  The query $\Phi'$ can be expressed in MSO.
\end{claim}

\begin{proof}
  We can clearly express in MSO that the tree is well-formed. Further, once the
  tree is asserted to be well-formed, we can quantify variables $x_1, \ldots,
  x_n$ that are assigned to the unique tree nodes that respectively have $x_1,
  \ldots, x_n$ in their label, and then evaluate the original query~$\Phi$.
\end{proof}

\subparagraph*{Tree automata.}
We use the notion of \emph{bottom-up
deterministic tree automata}:

\begin{definition}
  Let $\Gamma$ be an alphabet. A \emph{$\Gamma$-bottom-up deterministic tree automaton} (bDTA)
  consists of a finite set $Q$ of \emph{states}, an \emph{initial function}
  $\iota\colon \Gamma \to Q$, a \emph{transition function} $\delta\colon Q
  \times Q \times \Gamma \to Q$, and a subset $F \subseteq Q$ of \emph{final
  states}.
\end{definition}

We omit the standard definitions of what it means for a bDTA to \emph{accept} a
tree~\cite{tata}. We always assume that bDTAs are \emph{trimmed}, i.e., for
every state $q$, there is a tree accepted by~$A$ on which state~$q$ appears:
this can be enforced in linear time on~$A$ by removing useless states, provided
that the language of~$A$ is non-empty.

It is well-known~\cite{thatcher1968generalized} that the Boolean MSO query
$\Phi'$ on $\Gamma$-trees
can be translated to a $\Gamma$-bDTA $A$ which is \emph{equivalent} in the sense
that, for any $\Gamma$-tree $T$, the bDTA $A$ accepts $T$ iff $T$
satisfies~$\Phi'$. We will use $A$ to construct the circuit and prove
Proposition~\ref{prp:compil}. We will also need the notion of a
\emph{well-formed} bDTA:

\begin{definition}
  Let $\Lambda$ be a tree alphabet, $X$ be a set of variables, and let $\Gamma
  \colonequals \Lambda \times 2^X$. A $\Gamma$-bDTA is \emph{well-formed}
  if it accepts only well-formed trees.
\end{definition}

We make an immediate observation on such bDTAs:

\begin{claim}
  \label{clm:domain}
  Given a well-formed $\Gamma$-bDTA $A$ with state set~$Q$,
  writing $\Gamma = \Lambda \times 2^X$,
  there is a function $\dom\colon Q \to 2^X$ with the following property: for any
  $\Gamma$-tree $T$, letting $q$ be the state obtained at the root when
  evaluating $A$ on~$T$,
  then the subset of variables of~$X$ that occurs in the labels of the nodes
  of~$T$ is precisely~$\dom(q)$. Formally, we have $\dom(q) = \bigcup_{n \in T}
  \pi_2(n)$, where $\pi_2\colon T \to 2^X$ maps every node of~$T$ to the second
  component of its label. What is more, if $q\in Q$ is final then $\dom(q) = X$.
\end{claim}

Accordingly, the \emph{domain} $\dom(q)$ of a state~$q \in Q$ of such a bDTA is the subset
of variables to which it is sent by this function.

\begin{proof}[Proof of Claim~\ref{clm:domain}]
  As $A$ is trimmed, we know that for every state $q$ there is a $\Gamma$-tree
  $T$ accepted by~$A$ such that $q$ appears in the run of~$A$ on~$T$. As $A$
  is well-formed, we know that $T$ is well-formed. Let $T'$ be a subtree
  rooted at some node which is mapped to~$q$ in this run. Define $\dom(q)$ to be the set $Y$ of
  variables occurring in the labels of~$T'$.

  To show that this definition is consistent, let us assume by
  contradiction that there is another 
$\Gamma$-tree
  $T_2$ accepted by~$A$ such that $q$ appears in the run of~$A$ on~$T_2$, and
  let $T_2'$ be a subtree
  rooted at some node which is mapped to~$q$ in this run.
  Assume by contradiction that the set of variables occurring in the
  labels of~$T_2'$ is different, say $Y'$ with $Y' \neq Y$. Let $T_3$ be the
  tree obtained from~$T$ by replacing the subtree~$T'$ with $T'_2$.

  We claim
  that~$T_3$ is not well-formed. Indeed, as $T$ is well-formed and the
  variables occurring in the labels of~$T'$ are~$Y$, then the variables occurring
  in the labels of $T \setminus T'$ are $X \setminus Y$, and as $Y \neq Y'$, we
  know that either $X \setminus Y$ and $Y'$ are not disjoint (i.e., a variable
  occurs twice) or $(X \setminus Y) \cup Y' \neq X$ (i.e., a variable in
  missing). In both cases, $T_3$ is not well-formed. But $A$ accepts~$T_3$,
  because it maps the root of~$T_2'$ in~$T_3$ to~$q$ by hypothesis, and then the
  run can be completed like in~$T \setminus T'$ and $T_3$ is accepted.

  Thus,
  $A$ accepts the tree $T_3$ which is not well-formed, contradicting the
  assumption that~$A$ is well-formed.

  The last sentence of the claim is simply by observing that if $q$ is final
  then any tree $T$ where the root is mapped by~$A$ to~$q$ is accepted by~$A$, so as
  $A$ is well-formed it must be the case that~$T$ is well-formed so that $\dom(q)$, the set of
  all variables occurring in~$T$, must be~$X$.
\end{proof}

\subparagraph*{Provenance circuits for tree automata.}
We finally claim that we can construct provenance circuits:

\begin{proposition}
  \label{prp:circuit}
  Let $\Lambda$ be an alphabet, let $X$ be a non-empty set of variables, let $\Gamma \colonequals
  \Lambda \times 2^X$. Given a well-formed $\Gamma$-bDTA $A$ whose language is
  non-empty, given
  a $\Lambda$-tree $T$, we
  can build in time $O(|A| \times |T|)$ a smooth multivalued d-DNNF circuit $C$
  on domain $T$ and variables~$X$ such that $\rel{C}$ is precisely the set of
  assignments $\alpha$ such that $A$ accepts $T_\alpha$.
\end{proposition}

\begin{proof}
  In this proof, we will construct a variant of smooth multivalued d-DNNF
  circuits where we allow $\cup$-gates and $\times$-gates with no inputs.
  Intuitively, for a $\cup$-gate $g$ with no input we have $\rel{g} =
  \emptyset$, and for a $\times$-gate $g$ with no input we have $\rel{g} =
  \{[]\}$. One can show that, given such a circuit, we can rewrite it in
  linear time to a smooth multivalued d-DNNF without such gates, simply by
  eliminating all such gates bottom-up (see~\cite{amarilli2017circuit} for
  similar results)
  Note that, as the set of variables is non-empty and
  as the language of~$A$ is non-empty, the resulting smooth multivalued d-DNNF
  $D$ is such that $\rel{D} \neq \emptyset$ and $[] \notin \rel{D}$.

  Let $Q$ be the state set of~$A$, let $\iota$ be its initial function, let
  $\delta$ be its transition function, let $F$ be its set of final states.
  We use Claim~\ref{clm:domain} to define the function $\dom$ on~$Q$.
  We create one $\cup$-gate $g_{n,q}$ for every state~$q \in Q$ and node~$n\in T$, intuitively
  denoting that the automaton is at state~$q$ at node~$n$.

  For every leaf node $n$ of~$T$, letting $\lambda$ be the label of~$n$, for every subset $Y \subseteq X$, we create a
  $\times$-gate having as input the variables $\langle y: n\rangle$ for each $y
  \in Y$ (possibly none), with a wire to the gate $g_{n,q}$ where $q
  \colonequals \iota((\lambda, Y))$.

  For every internal node $n$ of~$T$ with children $n_1$ and $n_2$,
  letting~$\lambda$ be the label of~$n$, for any pair
  of states $q_1$ and $q_2$ such that $Y_1 \colonequals \dom(q_1)$ and
  $Y_2 \colonequals \dom(q_2)$ are disjoint,
  for any $Y \subseteq X \setminus (Y_1 \uplus Y_2)$, we create a $\times$-gate
  having as inputs $g_{n_1, q_1}$, $g_{n_2, q_2}$, and the variables $\langle y:
  n \rangle$ for each $y \in Y$ (possibly none), with a wire to the gate
  $g_{n,q}$ where $q \colonequals \delta(q_1, q_2, (\lambda, Y))$.

  Last, for the root node~$n_0$ of~$T$, we create a $\lor$-gate having as input
  all gates $g_{n_0,q}$ for $q \in F$, and set it to be the output gate.

  This construction runs in time $O(|A| \times |T|)$. We first check that the
  circuit is decomposable. For this, we show by an immediate induction that, for
  any $n\in T$ and $q\in Q$, then $\var{g_{n,q}} = \dom(q)$. In particular,
  $\dom(C) = X$. Decomposability is then clear. For smoothness, it is again easy
  to check that for every $n\in T$ and $q\in Q$, for every gate $g_{q,n}$, its
  inputs (if any) are $\times$-gates whose domain is $\dom(q)$. Further, for the
  output gate, its inputs are gates of the form $g_{n_0,q}$ where $q$ is final
  so $\dom(q) = X$.

  We next show an invariant: for every node $n$ of~$T$, for every state $q$
  of~$Q$, then $\rel{g_{n,q}}$ is precisely the set of partial assignments
  $\alpha$
  of~$X$ to the subtree $T_n$ of~$T$ rooted at~$n$ such that $A$ maps the root
  of~$(T_n)_\alpha$ to~$q$, where $(T_n)_\alpha$ denotes as expected the $\Gamma$-tree obtained from the $\Lambda$-tree $T_n$
  by setting the second component in~$2^X$ of every node $n'$ according to the
  subset of the variables of~$X$ that are mapped to~$n'$ by~$\alpha$.
  (Note that $(T_n)_\alpha$ is not necessarily well-formed because $\alpha$
  is a partial assignment.)

  The invariant is shown by
  induction. If $n$ is a leaf, letting $\lambda$ its label, then indeed for
  every $Y \subseteq 2^X$ the partial assignment mapping the variables of~$Y$
  to~$n$ is captured by the gate $g_{n, \iota(\lambda, Y)}$. If $n$ is an
  internal node with children $n_1$ and $n_2$ and label $\lambda$, we show both
  directions. First, consider a partial assignment $\alpha$ in $\rel{g_{n,q}}$
  for some~$q$. By construction of the circuit,
  it must witness that there are partial assignments
  $\alpha_1$ of~$Y_1$ to~$T_{n_1}$ in $\rel{g_{n_1, q_1}}$ and 
  $\alpha_2$ of~$Y_2$ to~$T_{n_2}$ in $\rel{g_{n_2, q_2}}$
  and $\alpha'$ of~$Y$ to~$n$, such that $Y_1$ and $Y_2$ and $Y$ are pairwise
  disjoint and $\alpha = \alpha_1 \times \alpha_2 \times \alpha$, and 
  such that $\delta(q_1, q_2, (\lambda, Y)) = q$.
  By induction hypothesis, this means that the automaton maps the root of
  $(T_{n_1})_{\alpha_1}$ to state~$q_1$ and  the root of
  $(T_{n_2})_{\alpha_2}$ to state~$q_2$, so that it maps the root of~$T$ to
  state~$q$. Conversely, consider a partial assignment~$\alpha$ such that the
  root of $(T_n)_\alpha$ is mapped by~$A$ to~$q$. Consider the restrictions
  $\alpha_1$, $\alpha_2$, and $\alpha'$ of~$\alpha$ to~$T_{n_1}$, $T_{n_2}$, and
  $n$ respectively. It must be the case that the root
  of~$(T_{n_1})_{\alpha_1}$ is mapped by~$A$ to some state~$q_1$, the root
  of~$(T_{n_2})_{\alpha_2}$ is mapped by~$A$ to some state~$q_2$, and
  we have $\delta(q_1, q_2, (\lambda, Y)) = q$, where $Y$ is the set of
  variables mapped by~$\alpha'$. Now, by induction hypothesis, we have 
  $\alpha_1 \in \rel{g_{n_1,q_1}}$ and
  $\alpha_2 \in \rel{g_{n_1,q_2}}$. Thus the construction of the circuit
  witnesses that $\rel{g_{n,q}}$ contains $\alpha$.

  The inductive claim implies that the circuit is deterministic. Indeed, the
  root gate is deterministic because, if the same $\alpha$ is captured by two
  different inputs $g_{n_0, q}$ and $g_{n_0, q'}$ of the root gate with $q \neq
  q'$, then the root of the tree $T_\alpha$ would be mapped both to~$q$ and~$q'$
  by~$A$, contradicting the determinism of~$A$. Further, for any $n \in T$ and
  $q\in Q$, we show that $g_{q,n}$ is deterministic. This is immediate
  if~$n$ is a leaf because the gates $g_{q,n}$ then have at most one
  input, so let us assume that $n$ is an internal node with
  children~$n_1$ and~$n_2$. Let us assume by
  contradiction that the same partial
  assignment $\alpha$ is in $\rel{g}$ and $\rel{g'}$ for two different inputs
  of~$g_{q,n}$. Let $Y$ be the set of variables mapped to~$n$ by~$\alpha$, let
  $Y_1$ and $Y_2$ be the set of variables mapped to nodes of $T_{n_1}$ and
  $T_{n_2}$ respectively. The sets $Y$ and $Y_1$ and $Y_2$ must be pairwise
  disjoint. By the consideration on the domains of gates, 
  the two inputs of~$g_{q,n}$ must be two $\times$-gates taking the
  conjunction of the variables $\langle y : n\rangle$ for $y \in Y$. If they are
  different gates, it must be the case that the two other gates that they
  conjoin are of the form $g_{n_1, q_1}$ and $g_{n_2, q_2}$ for the first gate,
  and $g_{n_1, q_1'}$ and $g_{n_2, q_2'}$  for the second gate, with $(q_1,q_2)
  \neq (q_1', q_2')$. But it must be the case by the invariant that the root of
  $(T_{n_1})_{\alpha_1}$  is mapped by~$A$ to both $q_1$ and $q_1'$, and
that the root of
  $(T_{n_2})_{\alpha_2}$  is mapped by~$A$ to both $q_2$ and $q_2'$. So $q_1 =
  q_1'$ and $q_2 = q_2'$, a contradiction.

  Last, the inductive claim applied to the root implies that $\rel{C}$ is the
  set of assignments $\alpha$ such that $A$ accepts $T_\alpha$, which is what we
  wanted to show.

  Thus we have established that the circuit is a smooth multivalued d-DNNF with
  the correct semantics, which concludes the proof.
\end{proof}

Finally, the proof of Proposition~\ref{prp:compil} follows by using
Claim~\ref{clm:phi2} to construct $\Phi'$, by constructing the equivalent
bDTA~$A$, and finally by using Proposition~\ref{prp:circuit} to build the circuit.

\end{toappendix}

\section{Conclusion}
\label{sec:conclusion}
We have studied the problem of ranked enumeration for tractable circuit classes
from knowledge compilation, namely, DNNFs and d-DNNFs,
in the setting of multivalued circuits so as to apply these
results to ranked enumeration for MSO query answers on trees.
We have shown that the latter task can be solved with linear-time
preprocessing and delay logarithmic in the number of answers produced so far, in
particular logarithmic delay in the input tree in data complexity.
This result on trees is the
analogue of a previous result on words~\cite{bourhis2021ranked}, achieving the
same bounds but for a different notion of ranking functions.

We leave several questions open for future work. For instance, our
efficient algorithms always assume that the input circuits are smooth:
although this can be ensured ``for free'' in the setting of MSO on trees, it
is generally quadratic to enforce on an arbitrary input circuit~\cite{shih2019smoothing}. It may be
possible to perform enumeration directly on non-smooth circuits, or on
implicitly smoothed circuits, e.g., with special gates as
in~\cite{amarilli2017circuit}.
It would also be natural to study this problem in combined complexity, or for
free second-order variables,
though our algorithms cannot work on the RAM model if we need to store a
superpolynomial number of assignments in memory.
Last, it may be
possible to extend our algorithms to more general ranking functions than the one
we study, for instance by leveraging the framework of MSO cost
functions used in~\cite{bourhis2021ranked}, or using weighted
logics~\cite{droste2005weighted}, or possibly replacing subset-monotonicity by
a weaker guarantee.

Last, it would be interesting to study whether
our results can extend to the support of \emph{updates}, e.g., reweighting
updates to the ranking functions, or updates on the underlying circuits or (for
MSO queries) on the tree, as in~\cite{losemann2014mso}
or~\cite{amarilli2019enumeration}. However, this is more difficult than the
case of updates for non-ranked enumeration, because our algorithms use larger
intermediate structures which are more challenging to maintain.

\bibliography{main}

\appendix

\end{document}